\documentclass{article}
\usepackage{graphicx}
\usepackage{amsmath,amssymb,mathtools}
\usepackage{booktabs}
\usepackage{makecell}
\usepackage{url}
\usepackage{mathptmx}
\usepackage{times}
\makeatletter
\let\cl@part\relax  
\makeatother
\usepackage{algorithm}        
\usepackage{algpseudocode}    
\makeatletter
\@ifpackageloaded{algpseudocode}{%

}{}
\makeatother
\newenvironment{proof}[1][Proof]{\par\noindent\textbf{#1.}\ }{\hfill$\square$\par}
\usepackage[numbers,sort&compress]{natbib}
\usepackage{aliascnt}
\makeatletter
\def\@begintheorem#1#2{%
  \trivlist
  \item[\hskip\labelsep{\bfseries #1\ #2}]\itshape}
\def\@opargbegintheorem#1#2#3{%
  \trivlist
  \item[\hskip\labelsep{\bfseries #1\ #2\ (#3)}]\itshape}
\makeatother

\newaliascnt{lemma}{theorem}
\newtheorem{lemma}[lemma]{Lemma}
\aliascntresetthe{lemma}

\newaliascnt{proposition}{theorem}
\newtheorem{proposition}[proposition]{Proposition}
\aliascntresetthe{proposition}

\newaliascnt{corollary}{theorem}

\aliascntresetthe{corollary}

\newaliascnt{conjecture}{theorem}

\aliascntresetthe{conjecture}

\newaliascnt{definition}{theorem}
\newtheorem{definition}[definition]{Definition}
\aliascntresetthe{definition}

\newaliascnt{remark}{theorem}
\newtheorem{remark}[remark]{Remark}
\aliascntresetthe{remark}

\newaliascnt{assumption}{theorem}
\newtheorem{assumption}[assumption]{Assumption}
\aliascntresetthe{assumption}

\newaliascnt{hypothesis}{theorem}

\aliascntresetthe{hypothesis}

\newaliascnt{property}{theorem}

\aliascntresetthe{property}

\usepackage[capitalize]{cleveref}
\Crefname{figure}{Fig.}{Figs.}
\crefname{lemma}{Lemma}{Lemmas}
\Crefname{lemma}{Lemma}{Lemmas}
\crefname{proposition}{Proposition}{Propositions}
\Crefname{proposition}{Proposition}{Propositions}
\crefname{corollary}{Corollary}{Corollaries}
\Crefname{corollary}{Corollary}{Corollaries}
\crefname{definition}{Definition}{Definitions}
\Crefname{definition}{Definition}{Definitions}
\crefname{remark}{Remark}{Remarks}
\Crefname{remark}{Remark}{Remarks}
\crefname{assumption}{Assumption}{Assumptions}
\Crefname{assumption}{Assumption}{Assumptions}
\crefname{conjecture}{Conjecture}{Conjectures}
\Crefname{conjecture}{Conjecture}{Conjectures}
\crefname{hypothesis}{Hypothesis}{Hypotheses}
\Crefname{hypothesis}{Hypothesis}{Hypotheses}
\crefname{property}{Property}{Properties}
\Crefname{property}{Property}{Properties}

\usepackage{mpcsymbols}

\crefformat{equation}{(#2#1#3)}
\providecommand{\runauthor}[1]{}%
\providecommand{\runtitle}[1]{}%
\providecommand{\thanksref}[1]{}%

\begin{document}

\title{Disturbance Attenuation Regulator I-B: Signal Bound Convergence and Steady-State\thanks{The authors gratefully acknowledge the financial support of the National Science Foundation (NSF) under Grant Nos. 2027091 and 2138985. The authors thank Moritz Diehl for helpful discussions.}}

\author{Davide Mannini\thanks{Department of Chemical Engineering, University of California, Santa Barbara. Email: dmannini@ucsb.edu} \and James B. Rawlings\thanks{Department of Chemical Engineering, University of California, Santa Barbara. Email: jbraw@ucsb.edu}}

\maketitle

\begin{abstract}
This paper establishes convergence and steady-state properties for the signal bound disturbance attenuation regulator (SiDAR). Building on the finite horizon recursive solution developed in a companion paper, we introduce the steady-state SiDAR and derive its tractable linear matrix inequality (LMI) with $O(n^3)$ complexity. Systems are classified as degenerate or nondegenerate based on steady-state solution properties. For nondegenerate systems, the finite horizon solution converges to the steady-state solution for all states as the horizon approaches infinity. For degenerate systems, convergence holds in one region of the state space, while a turnpike arises in the complementary region. When convergence holds, the optimal multiplier and control gain are obtained directly from the LMI solution. Numerical examples illustrate convergence behavior and turnpike phenomena.

Companion papers address the finite horizon SiDAR solution and the stage bound disturbance attenuation regulator (StDAR).
\end{abstract}

\section{Introduction}
\label{sec:intro}

The disturbance attenuation regulator (DAR) is a deterministic minmax feedback control for linear systems that attenuates the effect of exogenous bounded disturbances on a quadratic cost. The DAR is studied in two forms: the signal bound disturbance attenuation regulator (SiDAR), in which the disturbance signal two-norm is bounded over the entire horizon, and the stage bound disturbance attenuation regulator (StDAR), in which the disturbance is bounded independently at each stage. The term regulator denotes feedback regulation of the origin in the presence of disturbances.

In a companion paper \citep{mannini:rawlings:2026a}, we derived a finite horizon recursive solution to the SiDAR for arbitrary initial states, establishing that the optimal control policy is nonlinear in the state and partitions the state space into linear and nonlinear solution regions.

This paper extends those results to the infinite horizon setting. While the finite horizon solution is valid for any fixed horizon $N$, practical implementation requires understanding the limiting behavior as $N \to \infty$. Three fundamental questions arise: (i) does the finite horizon solution converge? (ii) what is the limiting solution? (iii) can the infinite horizon problem be solved efficiently?

We answer these questions by introducing the steady-state SiDAR, a fixed-point problem capturing the limiting behavior of the Riccati recursion. The steady-state problem admits a tractable LMI representation with $O(n^3)$ complexity; the optimal multiplier and control gain are then obtained directly from the LMI solution. For the zero initial state case, this LMI recovers the standard $H_\infty$ state feedback control problem.

The convergence analysis reveals a fundamental dichotomy. We classify systems as \emph{nondegenerate} or \emph{degenerate} based on whether the steady-state constraint $\lambda = \norm{G'\Pi G}$ is active at the origin. For nondegenerate systems, the finite horizon solution converges to the steady-state uniformly over the entire state space. For degenerate systems, convergence holds only in a subset of the state space; in the complementary region, the Riccati recursion exhibits turnpike behavior.

The main contributions are:
\begin{itemize}
\item introducing the steady-state SiDAR and deriving its LMI with $O(n^3)$ complexity
\item classifying systems as degenerate or nondegenerate based on steady-state properties
\item establishing convergence of the finite horizon solution to the steady-state as the horizon approaches infinity
\item investigating turnpike behavior for degenerate systems
\end{itemize}

The paper is organized as follows. \Cref{sec:prelim} summarizes the finite horizon SiDAR from \citep{mannini:rawlings:2026a}. \Cref{sec:ss} introduces the steady-state SiDAR and derives its LMI. \Cref{sec:inf} establishes convergence properties for degenerate and nondegenerate systems. \Cref{sec:num-example} illustrates the theory with numerical examples, and \Cref{sec:end} summarizes the findings. The appendix compiles fundamental propositions.

Companion papers address the finite horizon SiDAR solution \citep{mannini:rawlings:2026a} and the stage bound disturbance attenuation regulator (StDAR) \citep{mannini:rawlings:2026c}. A discussion of prior literature on the SiDAR and its connections to $H_\infty$ and minmax control is given in \citep{mannini:rawlings:2026a}.

\textit{Notation:} Let $\bbR $ denote the reals and $\bbI$ the integers. $\mathbb{R}^{m \times n}$ denotes the space of $m \times n$ real matrices and $\mathbb{S}^n$ denotes the space of $n \times n$ real symmetric positive definite matrices.
The \(\norm{x}\) denotes the two-norm of vector \(x\); \(\norm{M}\) denotes the induced two-norm of matrix \(M\); \(\normf{M}_F\) denotes the Frobenius norm of matrix \(M\). For matrices $X, Y \in \mathbb{R}^{m \times n}$, the matrix inner product is $\langle X, Y \rangle \eqbyd \operatorname{Tr}(X'Y)$, and $\normf{M}_F = \sqrt{\langle M, M \rangle}$. For a symmetric matrix $A \in \mathbb{R}^{n \times n}$ with $A \succeq 0$, $A^{1/2}$ denotes the unique principal symmetric square root: $A^{1/2} \succeq 0$ and $(A^{1/2})^2 = A$. For $A \succ 0$, define $A^{-1/2} \eqbyd (A^{1/2})^{-1}$. For a symmetric matrix $\Gamma \succeq 0$, we may write $\Gamma = WW'$ where $W \eqbyd \Gamma^{1/2}$ denotes the principal square root unless stated otherwise; in general, such factorizations are not unique. For a vector $y \in \mathbb{R}^p$, let $\yseq$ denote a sequence $\yseq \eqbyd (y_0, y_1, \dots, y_{N-1})$. The two-norm of a signal $\yseq$ is defined as $\smax{\yseq} \eqbyd ( \sum_{k=0}^{N-1} \norm{y_k}^2 )^{1/2}$ for finite horizon and $\smax{\yseq} \eqbyd ( \sum_{k=0}^{\infty} \norm{y_k}^2 )^{1/2}$ for infinite horizon; the one-norm of a sequence is defined as $\smax{\yseq}_1 \eqbyd \sum_{k=0}^{N-1} \norm{y_k}$. The column space (range) and nullspace of a matrix $M$ are denoted by $\mathcal{R}(M)$ and $\mathcal{N}(M)$, respectively. The pseudoinverse of a matrix $M$ is denoted as $M^{\dagger}$. The determinant of a square matrix $M$ is denoted $\det M$, and the adjugate (classical adjoint) is denoted $\mathrm{adj}(M)$. For symmetric matrices $A$ and $B$, $A \succeq B$ denotes $A - B$ is positive semidefinite (the Loewner order); a minimal solution refers to the smallest solution in the Loewner order. For a square matrix $M$, $\eig(M)$ denotes the set of eigenvalues of $M$ and $\rho(M)\eqbyd\max\{|\mu| : \mu\in\eig(M)\}$ denotes the spectral radius.

\section{Problem Setup and Finite Horizon Preliminaries}
\label{sec:prelim}

We summarize the essential definitions and results from the companion paper \citep{mannini:rawlings:2026a} required for the convergence analysis.

\subsection{Problem Set Up}

Consider the discrete time system
\begin{equation}
    x^+ = Ax + Bu + Gw \label{system}
\end{equation}  
in which $x \in \bbR^n$ is the state, $u \in \bbR^m$ is the control, $w \in \bbW \subset \bbR^q$ is a disturbance, and $x^+\in \bbR^n$ is the successor state. Define the control and disturbance sequences $\useq \eqbyd (u_0,u_1,\dots,u_{N-1})$ and $\wseq \eqbyd (w_0,w_1,\dots,w_{N-1})$. The signal bound disturbance constraint set is
\[
	\bbW
	\eqbyd
	\Bigl\{
	\wseq
	 \mid
	\sum_{k=0}^{N-1}|w_k|^2 \leq \alpha
	\Bigr\}
    \] 
The objective function is
\begin{equation}
    V(x_0, \useq, \wseq) = \sum_{k=0}^{N-1} \ell(x_k, u_k)  + \ell_f(x_N) \label{maincost}
\end{equation}
where $\ell(x,u) = (1/2)x'Qx + (1/2)u'Ru$ and $\ell_f(x) = (1/2)x'P_fx$ with $Q \succeq 0$, $R \succ 0$, and $P_f \succeq 0$. The state $x_k$ in \eqref{maincost} denotes the solution of \eqref{system} at stage $k$ generated from the initial state $x_0$ by the control and disturbance sequences $\useq$ and $\wseq$.

\begin{assumption}
$(A,B)$ stabilizable and $(A,Q)$ detectable. \label{asst1}
\end{assumption}
\begin{assumption}
$\mathcal{R}(G)\subseteq\mathcal{R}(B)$. \label{asst2}
\end{assumption}
Assumption~\ref{asst2} ensures every column of $G$ lies in the column space of $B$, which is a sufficient condition for nonsingularity of the block matrix appearing in the Riccati recursion (Proposition~24 of \citep{mannini:rawlings:2026a}). Relaxing it requires pseudoinverses in place of inverses and is left to future work.
\begin{assumption}
$G'P_fG\neq0$. \label{asst3}
\end{assumption}
\begin{assumption}
$Q\succ0$, $P_f \succ0$. \label{asst4}
\end{assumption}
The strict definiteness in Assumption~\ref{asst4} is not required for the finite horizon analysis; it is invoked starting from the steady-state problem in Section~\ref{sec:ss}.

The signal bound disturbance attenuation regulator (SiDAR) is obtained from the optimization problem
\begin{equation}
    V^*(x_0) \eqbyd \min_{u_0}\max_{w_0} \; \min_{u_1}\max_{w_1} \; \cdots \min_{u_{N-1}}\max_{w_{N-1}} \;
 \frac{V(x_0, \useq, \wseq)}{\sum^{N-1}_{k=0} |w_k|^2 } \quad \wseq \in \bbW\label{signaldp-w}
\end{equation}

\subsection{Finite Horizon Solution}

The cost in \eqref{signaldp-w} is written in terms of the sequences $\useq$ and $\wseq$, but the SiDAR is not an open-loop minimization over $\useq$. The problem is solved by backward dynamic programming, and the optimal solution is a state feedback policy $u^*_k(x_k)$. The policy is implemented online. At each stage, the scalar optimization \eqref{lsi} is solved for the measured state and the remaining disturbance budget to obtain the optimal multiplier and the control gain, as developed in Remarks~11 and~12 of the companion paper \citep{mannini:rawlings:2026a}. In the companion paper \citep{mannini:rawlings:2026a}, the original minmax SiDAR \eqref{signaldp-w} is reformulated by introducing a Lagrange multiplier $\lambda \geq 0$ for the signal bound constraint, and shown to reduce to a scalar convex optimization in $\lambda$. The optimal multiplier $\lambda^*(x_0)$ is the minimizer of this scalar program; the stagewise bounds $\lambda_k$ in \eqref{lsi} below are feasibility lower bounds of the backward Riccati recursion, not separate decision variables. The reduction of \eqref{signaldp-w} to \eqref{lsi} is derived by the stacked Lagrangian argument in the proof of Proposition~7 of \citep{mannini:rawlings:2026a}.

The following result from \citep{mannini:rawlings:2026a} provides the finite horizon solution.

\begin{proposition}[Finite horizon SiDAR {\citep[Proposition 7]{mannini:rawlings:2026a}}]
    \label{prop:ndsignal}
Let Assumptions 1-3 hold. Consider the following scalar convex optimization 
\begin{gather}
\mathbf{L}_{si}: \quad \min_{\lambda \in [\lambda_1, \infty)}
     \; \frac12\!\left(\frac{x_0}{\sqrt{\alpha}}\right)'\!
           \Pi_0(\lambda)\!\left(\frac{x_0}{\sqrt{\alpha}}\right)
     +\frac{\lambda}{2}\ \label{lsi} 
\end{gather}
\begin{gather*}
\lambda_N \eqbyd \norm{G'P_fG} \\
\begin{split}
\lambda_k \eqbyd
\begin{cases}
 \begin{aligned}
  &\min_{\lambda\ge \lambda_{k+1}} \Bigl\{\, \lambda : \lambda =  \norm{ G'\Pi_{k+1}(\lambda)G} \Bigr\}
 \end{aligned} \\
 \qquad \qquad \text{if }\;\norm{G'\Pi_{k+1}(\lambda_{k+1})G}>\lambda_{k+1}\\[10pt]
 \lambda_{k+1} \\
 \qquad \qquad \text{if }\;\norm{G'\Pi_{k+1}(\lambda_{k+1})G}\le\lambda_{k+1}
\end{cases}
\end{split}
\end{gather*}
    subject to the Riccati recursion
    \begin{equation}
        \Pi_k(\lambda) = Q+A'\Pi_{k+1}A-A' \Pi_{k+1} \begin{bmatrix} B & G \end{bmatrix}
     M_k(\lambda)^{-1}
    \begin{bmatrix} B' \\ G'\end{bmatrix}\Pi_{k+1}A \label{signalrec1}
    \end{equation}
    where
    \[
    M_k(\lambda) \eqbyd \begin{bmatrix} B'\Pi_{k+1}B + R & B'\Pi_{k+1} G \\ (B'\Pi_{k+1}G)' & G'\Pi_{k+1}G - \lambda I  \end{bmatrix}
    \]
    for $k \in [0,1,\dots,N-1]$ and terminal condition $\Pi_N = P_f \succeq0$. Given the solution to the \textit{scalar} convex optimization \eqref{lsi}, $\lambda^*(x_0)$, then
    \begin{enumerate}
        \item The optimal control policy $u^*_k(x_k;\lambda^*)$ to \eqref{signaldp-w} satisfies the stationary conditions
    \begin{equation}
     M_k(\lambda^*)
    \begin{bmatrix} u_k \\ z_k \end{bmatrix}^* = -\begin{bmatrix} B' \\ G'\end{bmatrix} \Pi_{k+1}A \; x_k \label{ndsc2}
    \end{equation}
    where $z_k$ denotes the Lagrangian stationary disturbance from the unconstrained stationary conditions, computed by the second block of \eqref{ndsc2}. The constrained optimal disturbance $w^*_k(x_k;\lambda^*)$ is given separately in item~2.
    \item The optimal disturbance policy $w^*_k(x_k;\lambda^*) = \overline{w}_k(u^*_k) \cap \bbW$ to \eqref{signaldp-w} satisfies
    \begin{equation}
    \begin{split}
    (B'\Pi_{k+1}G)'u^*_k(x_k;\lambda^*) &+(G'\Pi_{k+1}G - \lambda^* I ) \ \overline{w}_k(u^*_k) \\
    &= -G'\Pi_{k+1}Ax_k
    \end{split}
    \label{ndwcond2}
    \end{equation}
    where $\overline{w}_k(u^*_k) \subset \bbR^q$ is the set of solutions of \eqref{ndwcond2}.
    \item The optimal cost to \eqref{signaldp-w} is
    \begin{equation}
    V^*(x_0) = (1/2)\;(\frac{x_0}{\sqrt{\alpha}})' \Pi_0(\lambda^*) (\frac{x_0}{\sqrt{\alpha}})  + \lambda^* /2 \label{ndoc2}
    \end{equation} 
    \item For $\lambda\geq\lambda_1$, $\Pi_{k}(\lambda)$ is monotonic nonincreasing in $k$ and in $\lambda$.
    \end{enumerate}
\end{proposition}

\subsection{Solution Regions}

For fixed budget $\alpha$, the state space partitions into two regions.

\begin{definition}[Solution regions for SiDAR] \hfill
\begin{enumerate}
    \item  Region $\mathcal{X}_{L}(\alpha) \subseteq \bbR^n$ is the set of initial states $x_0$ for which $\lambda^*(x_0) = \lambda_1$ is optimal in problem $\mathbf{L}_{si}$ \eqref{lsi}.
    \item Region $\mathcal{X}_{NL}(\alpha) \subseteq \bbR^n$ is the set of initial states $x_0$ for which $\lambda^*(x_0) >  \lambda_1$ is optimal in problem $\mathbf{L}_{si}$ \eqref{lsi}.
\end{enumerate}
\end{definition}

From Proposition 15 in Mannini and Rawlings~\citep{mannini:rawlings:2026a}, the region $\mathcal{X}_{L}(\alpha)$ is an ellipsoid centered at the origin, and the optimal control is linear in $x$ within this region.

\section{Steady-state SiDAR}
\label{sec:ss}
This section introduces the steady-state SiDAR and derives its LMI representation via \cref{prop:lmigen}.

\subsection{Steady-state problem}

Let Assumptions 1-4 hold. Given the linear system \eqref{system}, we define the following optimization, denoted as steady‑state SiDAR
\begin{subequations}\label{eq:ss-problem}
\begin{align}
  \min_{\lambda,\Pi}\;&\; V(x_0,\lambda,\Pi)
  \label{eq:ss-problem:cost}\\
  \text{s.\,t. }& \;\lambda\geq\norm{G'\Pi G}\quad
                  g(\lambda,\Pi)=0
  \label{eq:ss-problem:constraints}
\end{align}
\end{subequations}
where
\[
V(x,\lambda,\Pi) \eqbyd (1/2) (\frac{x_0}{\sqrt{\alpha}})' \Pi (\frac{x_0}{\sqrt{\alpha}})  + \lambda /2 
\]
\begin{equation}
\begin{split}
g(\lambda,\Pi) &\eqbyd Q+A'\Pi A -A' \Pi \begin{bmatrix} B & G \end{bmatrix} \\
&\quad \begin{bmatrix} B'\Pi B + R & B'\Pi G \\ (B'\Pi G)' & G'\Pi G - \lambda I  \end{bmatrix}^{-1}
    \begin{bmatrix} B' \\ G'\end{bmatrix}\Pi A - \Pi
\end{split}
\label{gareinv}
\end{equation}
\[
  M(\lambda,\Pi)\eqbyd
  \begin{bmatrix}
      B' \Pi B + R & B' \Pi G\\
      G' \Pi B     & G' \Pi G - \lambda I
  \end{bmatrix}
\]
Let $(\overline{\lambda},\overline{\Pi})$ denote the solution to \eqref{eq:ss-problem}. We define two solution regions for the steady-state problem \eqref{eq:ss-problem}.
\begin{definition}[Solution regions for \eqref{eq:ss-problem}] \hfill
\begin{enumerate}
    \item  Region $\overline{\mathcal{X}}_{L}(\alpha) \subseteq \bbR^n$ is the set of states $x_0$ for which $\overline{\lambda} = \norm{G'\overline{\Pi}G}$ is optimal in problem \eqref{eq:ss-problem}.
    \item Region $\overline{\mathcal{X}}_{NL}(\alpha) \subseteq \bbR^n$ is the set of states $x_0$ for which $\overline{\lambda} > \norm{G'\overline{\Pi}G}$ is optimal in problem \eqref{eq:ss-problem}.
\end{enumerate}
\end{definition}
\cref{prop:lmigen} establishes an LMI for the steady-state SiDAR \eqref{eq:ss-problem}. Existence of a solution to the LMI is addressed in \cref{prop:lmi-existence}.

\begin{proposition}[LMI for steady-state SiDAR]
    \label{prop:lmigen}
The solution to steady‑state SiDAR \eqref{eq:ss-problem} is implied by the following optimization
\begin{equation}
        \min_{\lambda,\chi,P,F} \; \lambda /2 + \chi/2 \label{inflmi}
    \end{equation}
    subject to
\begin{gather*}
\begin{bmatrix}
   P & (AP - BF)' & 0 & P\hat{Q}' - F'\hat{R}'\\
   AP - BF & P & G & 0 \\
   0 & G' & \lambda I & 0\\
   (P\hat{Q}' - F'\hat{R}')'  & 0 & 0 & I
\end{bmatrix} \succeq 0 \\
\begin{bmatrix}
   P & \frac{x_0}{\sqrt{\alpha}}\\
   \frac{x_0'}{\sqrt{\alpha}} & \chi 
\end{bmatrix}\succeq 0 
\end{gather*}    
    where $K = -F P^{-1}$, $P = \Pi^{-1} \succ 0$, $\hat{Q}' = \begin{bmatrix} Q^{1/2} & 0 \end{bmatrix}$, and $\hat{R}' = \begin{bmatrix} 0 & R^{1/2} \end{bmatrix}$.
\end{proposition}

\noindent\textbf{Proof sketch.}\ Under \cref{asst1}, the optimization \eqref{inflmi} admits a strictly feasible primal point, so strong duality holds between \eqref{inflmi} and its dual. Stationarity, primal feasibility, and complementary slackness then yield the steady-state Riccati identity \eqref{eq:ss-problem:constraints}. The full proof is in the appendix (\cref{app:proofs}).

\Cref{prop:lmi-existence} establishes that the LMI optimization~\eqref{inflmi} admits an optimal solution under \cref{asst1,asst2,asst3,asst4}.

\begin{proposition}[Existence of LMI solution]
\label{prop:lmi-existence}
Let Assumptions~\ref{asst1}--\ref{asst4} hold. The optimization problem \eqref{inflmi} has an optimal solution $(\lambda^*, \chi^*, P^*, F^*)$.
\end{proposition}

\noindent\textbf{Proof sketch.}\ The objective $\lambda/2 + \chi/2$ is continuous and bounded below on the closed feasibility set of \eqref{inflmi}. The LMI constraints bound $\lambda$, $\chi$, and $P$ on every level set of the objective, so each level set is compact. The Weierstrass theorem then gives existence. The full proof is in the appendix (\cref{app:proofs}).

\subsection{System Classes}
\label{subsec:sys-classes}

We classify the linear system \eqref{system} based on the steady-state solution properties.
\begin{definition}[steady-state system classes]\hfill
\begin{enumerate}
\item A system \eqref{system} is \textit{nondegenerate} if the solution $(\overline{\lambda},\overline{\Pi})$ to the steady-state problem \eqref{eq:ss-problem} satisfies $\overline{\lambda} = \norm{G'\overline{\Pi}G}$ for $x_0 = 0$.
\item A system \eqref{system} is \textit{degenerate} if the solution $(\overline{\lambda},\overline{\Pi})$ to the steady-state problem \eqref{eq:ss-problem} satisfies $\overline{\lambda} > \norm{G'\overline{\Pi}G}$ for all $x_0 \in \bbR^n$.
\end{enumerate}
\end{definition}

\begin{remark}
For nondegenerate systems, the steady-state feasible region $\overline{\mathcal{X}}_L(\alpha)$ contains at least the origin. For degenerate systems the region $\overline{\mathcal{X}}_L(\alpha)$ is empty, $\overline{\mathcal{X}}_L(\alpha) = \varnothing$. 
\end{remark}

\section{Infinite Horizon SiDAR Solution}
\label{sec:inf}
This section analyzes the limiting behavior of the finite horizon SiDAR as $N \to \infty$. We establish conditions under which the finite horizon solution converges to the steady-state solution from \Cref{sec:ss} and identify turnpike behavior in degenerate systems where convergence holds only in a subset of the state space. Throughout this section, the disturbance budget $\alpha$ is fixed.

Define $\lambda^*(N)$ as the optimal solution to $\mathbf{L}_{si}$ \eqref{lsi} for horizon $N \in \mathbb{I}_{\geq 1}$, and $\Pi_{0}(\lambda^*(N))$ as the optimal initial stage matrix generated by the backward recursion \eqref{signalrec1} for horizon $N$ with terminal condition $\Pi_{N} = P_f$.

For all nondegenerate systems, the following result establishes convergence of the finite horizon solution to the SiDAR \eqref{signaldp-w} as $N \rightarrow \infty$ to a minimal steady-state solution $(\overline{\lambda}, \overline{\Pi})$, where minimal indicates the Loewner order on $\Pi$.

\begin{proposition}[Convergence: nondegenerate systems]
\label{infl}
Let Assumptions 1-4 hold. For every nondegenerate system, the finite horizon solution to the SiDAR \eqref{signaldp-w} converges to a minimal steady-state solution $(\overline{\lambda}, \overline{\Pi})$
\[
  \lim_{N \to \infty} \lambda^*(N) = \overline{\lambda} \quad \text{and} \quad \lim_{N \to \infty} \Pi_{0}(\lambda^*(N)) = \overline{\Pi}
\]
\end{proposition}

\noindent\textbf{Proof sketch.}\ For nondegenerate systems, the iterates $\Pi_0(\lambda^*(N))$ are monotone in $N$ and bounded above by $\overline{\Pi}$, and the multipliers $\lambda^*(N)$ are bounded above by $\overline{\lambda}$ via a stationary suboptimal control. Monotone convergence yields the limits $\lambda^*(N) \to \overline{\lambda}$ and $\Pi_0(\lambda^*(N)) \to \overline{\Pi}$; the steady-state identity \eqref{eq:ss-problem:constraints} extends to the limit by continuity. The full proof is in the appendix (\cref{app:proofs}).

For fixed budget $\alpha$, define the limit set
\[
  \mathcal{X}_L^\infty(\alpha)\;\eqbyd\;\bigcap_{N=1}^{\infty}\mathcal{X}_L(\alpha, N)
\]
and 
\[
\mathcal{X}^\infty_{NL}(\alpha) = \bbR^n \setminus \mathcal{X}^\infty_{L}(\alpha)
\]
where $\mathcal{X}_L(\alpha, N)$ is the region $\mathcal{X}_L(\alpha)$ for horizon $N \in \mathbb{I}_{\geq 1}$.

\begin{proposition}[Well-defined limit regions]
\label{prop:limit-regions}
Let $\mathcal{X}_L(\alpha, N)$ be the region from by Proposition 15 in Mannini and Rawlings~\citep{mannini:rawlings:2026a} for fixed $\alpha$ and horizon $N\in\mathbb{I}_{\ge 1}$.
Then $\mathcal{X}_L^\infty(\alpha)$ is nonempty, closed, and convex, and hence $\mathcal{X}_{NL}^\infty(\alpha)$ is well-defined and open.
\end{proposition}

\begin{proof}
For each $N$, Proposition 15 in Mannini and Rawlings~\citep{mannini:rawlings:2026a} gives
\[
\mathcal{X}_L(\alpha, N)=\Bigl\{x_0\in\mathbb{R}^n:
\bigl(\tfrac{x_0}{\sqrt{\alpha}}\bigr)'
\tilde J_N'(\lambda_1(N))\tilde J_N(\lambda_1(N))
\bigl(\tfrac{x_0}{\sqrt{\alpha}}\bigr)\le 1\Bigr\}
\]
with $\tilde J_N'(\lambda_1(N))\tilde J_N(\lambda_1(N))\succeq0$, where $\tilde J_N(\lambda_1(N))$ is the finite horizon disturbance gain defined in equation~(19) of \citep{mannini:rawlings:2026a} evaluated at $\lambda = \lambda_1(N)$. Hence each $\mathcal{X}_L(\alpha, N)$ is a closed, convex ellipsoid (possibly degenerate). Moreover, $0\in\mathcal{X}_L(\alpha, N)$ for all $N$, so
$0\in\bigcap_{N=1}^{\infty}\mathcal{X}_L(\alpha, N)=\mathcal{X}_L^\infty(\alpha)$, proving nonemptiness. Intersections of closed and convex sets are closed and convex, so $\mathcal{X}_L^\infty(\alpha)$ is closed and convex. Therefore $\mathcal{X}_{NL}^\infty(\alpha)=\mathbb{R}^n\setminus\mathcal{X}_L^\infty(\alpha)$ is open.
\end{proof}

\begin{remark}[Membership test for $\mathcal{X}_L^\infty(\alpha)$]
\label{rem:limit-membership}
For $x_0 \in \bbR^n$, $x_0 \in \mathcal{X}_L^\infty(\alpha)$ if and only if
\[
(x_0/\sqrt{\alpha})' \tilde J_N'(\lambda_1(N)) \tilde J_N(\lambda_1(N)) (x_0/\sqrt{\alpha}) \leq 1
\]
for all $N \in \mathbb{I}_{\geq 1}$.
\end{remark}


\Cref{infld} establishes the limiting behavior of the finite horizon SiDAR for degenerate systems: the multiplier sequence converges, the $\mathcal{X}_L^\infty(\alpha)$ branch exhibits turnpike behavior with limit $(\lambda_\infty^*, \Pi_\infty(\lambda_\infty^*)) \neq (\overline{\lambda}, \overline{\Pi})$, and the $\mathcal{X}_{NL}^\infty(\alpha)$ branch converges to the steady-state solution.

\begin{proposition}[Convergence: degenerate systems]
\label{infld}
Let Assumptions 1-4 hold. For every degenerate system we have that
\begin{enumerate} 
\item the finite horizon solution to the SiDAR \eqref{signaldp-w} converges to a limit
\[
  \lim_{N \to \infty} \lambda^*(N) = \lambda^*_\infty \quad \text{and} \quad \lim_{N \to \infty} \Pi_{0}(\lambda^*(N)) = \Pi_\infty(\lambda^*_\infty)
\]
\item For $x_0\in\mathcal{X}^\infty_{L}(\alpha)$, we have $(\lambda^*_\infty,\Pi_\infty(\lambda^*_\infty))\neq(\overline{\lambda},\overline{\Pi})$ and $g(\lambda^*_\infty,\Pi_{\infty})>0$.
\item For $x_0\in\mathcal{X}^\infty_{NL}(\alpha)$ the finite horizon solution converges to a minimal steady-state solution $(\overline{\lambda}, \overline{\Pi})$
\[
  \lim_{N \to \infty} \lambda^*(N) = \overline{\lambda} \quad \text{and} \quad \lim_{N \to \infty} \Pi_{0}(\lambda^*(N)) = \overline{\Pi}
\]
\end{enumerate}
\end{proposition}

\noindent\textbf{Proof sketch.}\ For degenerate systems, $x_0 \in \mathcal{X}_L^\infty(\alpha)$ implies $\lambda^*(N) = \lambda_1(N)$ for all $N$, so $\lambda^*(N)$ converges to $\lambda_\infty^* \neq \overline{\lambda}$ and $g(\lambda_\infty^*, \Pi_\infty) > 0$ identifies the turnpike. For $x_0 \in \mathcal{X}_{NL}^\infty(\alpha)$, the argument of \cref{infl} applies and $\lambda^*(N) \to \overline{\lambda}$, $\Pi_0(\lambda^*(N)) \to \overline{\Pi}$. The full proof is in the appendix (\cref{app:proofs}).

\begin{remark}[LMI solution for infinite horizon SiDAR]
From \cref{infl,infld,prop:lmigen}, the infinite horizon SiDAR reduces to the tractable LMI optimization \eqref{inflmi} for
\begin{itemize}
    \item nondegenerate systems for $x_0\in\bbR^n$.
    \item degenerate systems for $x_0\in\mathcal{X}^\infty_{NL}(\alpha)$.
\end{itemize}
For degenerate systems where $x_0\in\mathcal{X}^{\infty}_L(\alpha)$, the infinite horizon solution differs from the steady-state solution obtained by \eqref{inflmi}.
\end{remark}

\begin{remark}[Connection to standard $H_\infty$ control]
For $x_0 = 0$ and nondegenerate systems, the LMI optimization \eqref{inflmi} reduces to the standard $H_\infty$ state feedback control LMI~\citep[p.128]{caverly:forbes:2019}, which is recovered as a special case of the infinite horizon SiDAR.
\end{remark}

\begin{remark}[Turnpike behavior for degenerate systems]
\label{rem:turnpike}
For degenerate systems with $x_0 \in \mathcal{X}_L^\infty(\alpha)$, \cref{infld} establishes that $g(\lambda^*_\infty, \Pi_\infty) > 0$, indicating the backward Riccati recursion does not settle at a steady-state fixed point. Consequently, for sufficiently large horizon $N$, the sequence $\Pi_k(\lambda^*)$ from recursion \eqref{signalrec1} exhibits a turnpike: a long interior plateau near $\Pi_\infty$ with short boundary layers at $k=0$ and $k=N$. This turnpike persists as $N \to \infty$ and is illustrated numerically in \cref{sec:num-example} (see \cref{fig:pirecursion}).
\end{remark}


\begin{remark}[Infinite horizon implementation via LMI]
\label{rem:receding-horizon}
At each stage $k$ with current state $x_k$ and remaining budget $b_k = \alpha - \sum_{j=0}^{k-1}|w_j|^2$, the infinite horizon steady-state LMI \eqref{inflmi} can be resolved by replacing $x_0$ with $x_k$ and $\alpha$ with $b_k$. This yields the optimal multiplier $\lambda^*(x_k, b_k)$ for the infinite horizon problem from the current state, from which the gain $K(\lambda^*(x_k, b_k))$ is computed. The optimal control policy requires resolving the steady-state LMI at each stage, as the realized state and disturbance history determine the current state and remaining budget for the infinite horizon optimization.
\end{remark}

\section{Numerical Examples}
\label{sec:num-example}
The following examples illustrate the convergence properties and steady-state behavior of the SiDAR. All analytical results hold for arbitrary dimension $n$; the scalar cases facilitate visualization. Throughout this section, we fix $\alpha = 1$, unless otherwise stated, and simplify the notation $\mathcal{X}_L(\alpha, N)$ to $\mathcal{X}_L(N)$.

\subsection{Degenerate vs Nondegenerate Systems: Scalar Cases}
Consider three scalar systems
\begin{alignat*}{7}
&\text{System 1.}\; \; & A&=1 &\; \; B&=1 &\; \; G&=1 &\; \; R&=1 &\; \; Q&=1 &\; \; & \text{\citep[p.\,92]{basar:bernhard:1995}}\\
&\text{System 2.}\; \; & A&=0.5 &\; \; B&=1 &\; \; G&=1 &\; \; R&=1 &\; \; Q&=0.2 &\; \; & \text{\citep[p.\,92]{basar:bernhard:1995}}\\
&\text{System 3.}\; \; & A&=0.5 &\; \; B&=0 &\; \; G&=1 &\; \; R&=1 &\; \; Q&=0.2 &\; \; &
\end{alignat*}
with terminal penalties $P_f = 1.8$ for system~1 and $P_f = 0.25$ for systems~2 and~3 (system~3 uses $P_f = 0.1$ in Figure~\ref{fig:xl}). The steady-state solutions to \eqref{eq:ss-problem} with $x_0=0$ yield
\begin{alignat*}{5}
&\text{System 1.}\; \; & \overline{\lambda}&=2 &\; \; \overline{\Pi}&=2 &\; \; \overline{\lambda}&=\norm{G'\overline{\Pi} G} &\; \; \; \; \overline{\mathcal{X}}_L&=[-2, 2]\\
&\text{System 2.}\; \; & \overline{\lambda}&=0.444 &\; \; \overline{\Pi}&=0.4 &\; \; \overline{\lambda}&>\norm{G'\overline{\Pi} G} &\; \; \; \; \overline{\mathcal{X}}_L&=\varnothing\\
&\text{System 3.}\; \; & \overline{\lambda}&=0.8 &\; \; \overline{\Pi}&=0.4 &\; \; \overline{\lambda}&>\norm{G'\overline{\Pi} G} &\; \; \; \; \overline{\mathcal{X}}_L&=\varnothing
\end{alignat*}
System~1 is nondegenerate; systems~2 and~3 are degenerate. System~3 violates Assumption~\ref{asst2} since $\mathcal{R}(G) = \mathbb{R} \not\subseteq \{0\} = \mathcal{R}(B)$, hence the theoretical results do not apply.

\begin{remark}
Systems with $\overline{\lambda} > \norm{G'\overline{\Pi}G}$ in the steady-state solution \eqref{eq:ss-problem} are degenerate because the steady-state feasible set $\overline{\mathcal{X}}_L$ is empty. However, the finite horizon solution satisfies $\mathcal{X}_L(N) \supseteq \{0\}$ for all $N$ by Proposition 15 in Mannini and Rawlings~\citep{mannini:rawlings:2026a}, thus the finite horizon regions cannot converge to an empty set.
\end{remark}

\Cref{fig:xl} depicts the $\mathcal{X}_L(N)$ regions for the SiDAR \eqref{signaldp-w} for systems~1–3 as functions of horizon length $N$ with $x_0=0$. System~1 exhibits rapid convergence to a steady-state region containing an open ball around the origin. For the degenerate system~2, the regions contract as $N$ increases but never become empty in finite time, consistent with Proposition 15 in Mannini and Rawlings~\citep{mannini:rawlings:2026a}, which guarantees $0 \in \mathcal{X}_L(N)$ for all $N$. System~3 exhibits $\mathcal{X}_L(N) = \{0\}$ for all horizon lengths, illustrating the pathological behavior when Assumption~\ref{asst2} fails.

\begin{figure}
\centering
\includegraphics[width=0.8\linewidth]{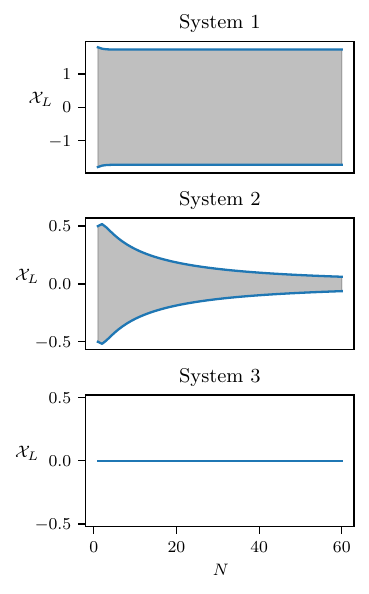}
\caption{$\mathcal{X}_L(N)$ regions versus horizon length $N$ for systems~1–3: nondegenerate system~1 (top), degenerate system~2 (middle), system~3 violating Assumption 2 (bottom).}
\label{fig:xl}
\end{figure}

\Cref{fig:pirecursion} shows the optimal recursion $\Pi_{k}(N)$ for systems~1–3 with varying horizon lengths $N \in \{50, 100, 150, 250\}$. The left panels report $x_0=0\in\mathcal{X}_L(N)$; the right panels report $x_0=2\in\mathcal{X}_{NL}(N)$. System~1 converges rapidly to the steady-state value $\overline{\Pi}=2$ for all $k$ and both initial conditions, consistent with \Cref{infl}. For the degenerate system~2, a turnpike is visible for $x_0=0\in\mathcal{X}_L(N)$: the recursion approaches but does not settle at a fixed point, exhibiting a long interior plateau with boundary layers at $k=0$ and $k=N$, as described in \Cref{rem:turnpike}. No turnpike is observed when $x_0=2\in\mathcal{X}_{NL}(N)$, where convergence to the steady-state holds by \Cref{infld}. System~3 illustrates breakdown of monotonicity in the final step from $k=1$ to $k=0$ for $x_0=0$, violating the theory due to failure of Assumption~\ref{asst2}.

\begin{figure}
\centering
\includegraphics[width=1.05\linewidth]{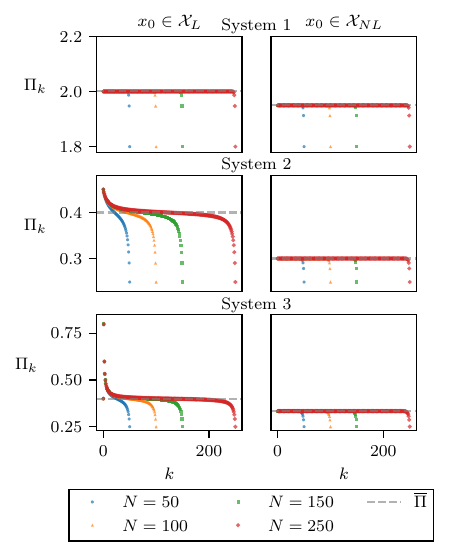}
\caption{Recursion $\Pi_{k}(N)$ versus stage $k$ for systems~1–3 with varying horizon lengths. Left: $x_0=0$. Right: $x_0=2\in\mathcal{X}_{NL}(N)$. Top: nondegenerate system~1; Middle: degenerate system~2; Bottom: system~3 violating Assumption~\ref{asst2}. Dashed lines indicate steady-state value $\overline{\Pi}$.}
\label{fig:pirecursion}
\end{figure}

\Cref{fig:lambdahorizon} shows the optimal solution $\lambda^*(N)$ as a function of horizon length $N$. The left column displays results for $x_0 =0\in \mathcal{X}_L$ and the right column for $x_0=2 \in \mathcal{X}_{NL}$. System~1 exhibits rapid convergence to the steady-state value $\overline{\lambda}$ for both initial conditions, as expected from \Cref{infl}. For the degenerate system~2, convergence to $\overline{\lambda}$ is guaranteed by \Cref{infld} for $x_0 =2\in \mathcal{X}_{NL}$; convergence for $x_0 =0\in \mathcal{X}_L$ is observed empirically but not guaranteed by the theory since the limit $(\lambda^*_\infty, \Pi_\infty)$ differs from the steady-state solution $(\overline{\lambda}, \overline{\Pi})$. System~3 converges empirically for both initial conditions, but this convergence is not guaranteed since Assumption~\ref{asst2} fails.

\begin{figure}
\centering
\includegraphics[width=1.02\linewidth]{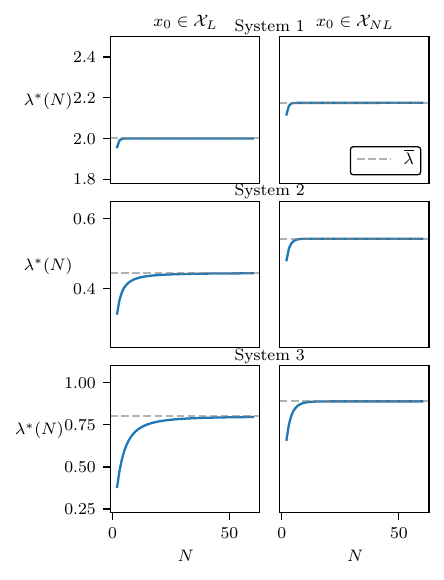}
\caption{Optimal solution $\lambda^*(N)$ versus horizon length $N$ for systems~1–3. Left column: $x_0 \in \mathcal{X}_L$. Right column: $x_0 \in \mathcal{X}_{NL}$. Rows: nondegenerate system~1 (top), degenerate system~2 (middle), system~3 violating Assumption~\ref{asst2} (bottom). Dashed lines indicate steady-state value $\overline{\lambda}$.}
\label{fig:lambdahorizon}
\end{figure}

\subsection{Degenerate vs Nondegenerate Systems: Multivariable Cases}
Consider two $2 \times 2$ systems
\begin{alignat}{4}
&\text{System 4:}\quad
& A&=\begin{bmatrix} 1 & 2 \\ 0 & 1 \end{bmatrix}
&\quad B&=\begin{bmatrix} 1 & 0 \\ 1 & 1 \end{bmatrix}
&\quad G&=Q=R=I \\
&\text{System 5:}\quad
& A&=\begin{bmatrix} 1 & 1 \\ 0 & 0.5 \end{bmatrix}
&\quad B&=\begin{bmatrix} 1 & 0 \\ 0 & 0.25 \end{bmatrix}
&\quad G&=Q=R=I
\end{alignat}
with terminal penalty $P_f = 0.01\,I$. System~4 is nondegenerate; system~5 is degenerate.

\Cref{fig:2dregions} illustrates the $\mathcal{X}_L(N)$ regions for horizon lengths $N \in \{3, 4, 10, 25\}$. The nondegenerate system rapidly converges to a fixed ellipse $\overline{\mathcal{X}}_L$ containing an open neighborhood of the origin, consistent with \Cref{infl}. For the degenerate system, the elliptical regions progressively contract as $N$ increases, asymptotically approaching a lower-dimensional set consistent with $\mathcal{X}^\infty_L$ from \Cref{prop:limit-regions}, while maintaining the origin in their interior for all finite $N$.

\begin{figure}
\centering
\includegraphics[width=0.95\linewidth]{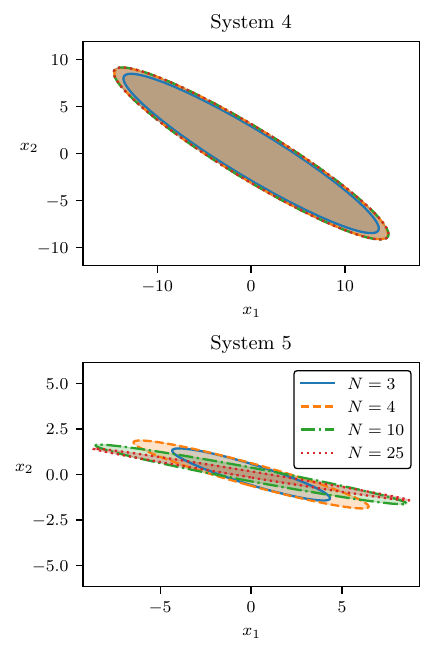}
\caption{$\mathcal{X}_L(N)$ regions for $2 \times 2$ systems with varying horizon lengths: nondegenerate system~4 (top), degenerate system~5 (bottom).}
\label{fig:2dregions}
\end{figure}

\subsection{Computational Complexity}
The LMI optimization \eqref{inflmi} from \Cref{prop:lmigen} was benchmarked across dimensions $n = m \in \{2, 6, 10, \ldots, 78\}$ using randomly generated stable nondegenerate systems. For each dimension, four independent samples were solved using MOSEK~\citep{mosek:2025}. The empirical complexity analysis reveals $O(n) \sim n^{3}$, consistent with interior point methods for semidefinite programming, confirming the computational tractability of the infinite horizon SiDAR via LMI.

\begin{remark}[Finite horizon scaling]
\label{rem:finite-scaling}
A complementary scaling study for the finite horizon SiDAR online algorithm is reported in \citet[Section~5.2]{mannini:rawlings:2026a}, where execution time is measured at state dimensions $n \in \{1, 2, 5, 10, 20, 50, 100, 200, 300\}$ for three combinations of $(m, q)$, including $(n, 1)$ and $(\lceil n/2 \rceil, \lceil n/2 \rceil)$.
\end{remark}

\section{Summary}
\label{sec:end}
This paper establishes convergence and steady-state properties for the SiDAR. Building on the finite horizon recursive solution from \citep{mannini:rawlings:2026a}, we introduced the steady-state SiDAR and derived its tractable LMI.

We classify linear systems as degenerate or nondegenerate based on whether the steady-state constraint $\lambda = \norm{G'\Pi G}$ is active at the origin. For nondegenerate systems, the finite horizon solution converges to the steady-state solution for all states as the horizon approaches infinity. For degenerate systems, convergence to the steady-state solution holds in one region of the state space ($\mathcal{X}_{NL}^\infty$), while a turnpike arises in the complementary region ($\mathcal{X}_L^\infty$).

When convergence holds, the infinite horizon SiDAR reduces to a tractable LMI optimization with $O(n^3)$ complexity; the optimal multiplier and control gain are then obtained directly from the LMI solution. At the origin, this recovers the standard $H_{\infty}$ state feedback control.

Evaluating the closed-loop state and control trajectories of the policy on specific applications, complementing the theoretical properties established here, is a subject for further study.

\Cref{tab:summary} summarizes the key convergence properties for nondegenerate and degenerate systems.

\begin{table*}
\refstepcounter{table}
\centering
\textbf{Table~\thetable.} Convergence properties of the finite horizon SiDAR as $N \to \infty$
\label{tab:summary}
\vspace{1ex}
\begin{tabular}{lcc}
\toprule
System & Region $\mathcal{X}_L^\infty$ & Region $\mathcal{X}_{NL}^\infty$ \\
\midrule
Nondegenerate & Converges to steady-state & Converges to steady-state \\
Degenerate    & Turnpike behavior & Converges to steady-state \\
\bottomrule
\end{tabular}
\end{table*}

\section{Appendix}
\label{sec:props}

In this appendix, we compile the fundamental results used throughout this paper.

The following lemma establishes equivalent forms of the Riccati recursion, expressing the value matrix $\Pi(\lambda)$ in terms of closed-loop quantities. This representation is used in the monotonicity analysis of \Cref{prop:ndsignal}.

\begin{lemma}[Riccati equalities]
\label{receq}
The equality 
\begin{align*}
\Pi(\lambda) &= Q+A'\Pi A- A' \Pi \begin{bmatrix} B & G \end{bmatrix} \\
&\quad \begin{bmatrix} B'\Pi B + R & B'\Pi G \\ (B'\Pi G)' & G'\Pi G - \lambda I  \end{bmatrix}^{\dagger}
    \begin{bmatrix} B' \\ G'\end{bmatrix}\Pi A
\end{align*}
    can be rewritten as
\begin{equation}
            \Pi = \bar{Q} + \bar{A}' \Pi\bar{A} - \bar{A}'\Pi G(G'\Pi G-\lambda I)^{\dagger}G'\Pi\bar{A} \label{xlrec}
        \end{equation}
    where $\bar{A} = A+BK$ and $\bar{Q} = Q + K'RK$ and $K$ satisfies
    \begin{equation}
  \begin{bmatrix}
        B'\Pi B+R & B'\Pi G \\
        G'\Pi B & G'\Pi G-\lambda I
    \end{bmatrix}\begin{bmatrix}
K \\ J
\end{bmatrix} =-\begin{bmatrix}
B'\Pi A \\ G'\Pi A
\end{bmatrix} \label{JK}
\end{equation}
\end{lemma}

\begin{proof}
From $M^\dagger M M^\dagger=M^ \dagger$ we have
 \begin{align}
         \Pi(\lambda) &=Q+A'\Pi A- A' \Pi \begin{bmatrix} B & G \end{bmatrix}
     M(\lambda)^{\dagger}
    \begin{bmatrix} B' \\ G'\end{bmatrix}\Pi A  \label{recjk} \\ &= Q+A'\Pi A- A' \Pi \begin{bmatrix} B & G \end{bmatrix}
     M(\lambda)^{\dagger}M(\lambda)M(\lambda)^{\dagger}
    \begin{bmatrix} B' \\ G'\end{bmatrix}\Pi A \nonumber
    \end{align}
where
\[
M(\lambda) = \begin{bmatrix} B'\Pi B + R & B'\Pi G \\ (B'\Pi G)' & G'\Pi G - \lambda I  \end{bmatrix}
\]
Define
\[
\begin{bmatrix}
        B'\Pi B+R & B'\Pi G \\
        G'\Pi B & G'\Pi G-\lambda I
    \end{bmatrix}\begin{bmatrix}
K \\ J
\end{bmatrix} =-\begin{bmatrix}
B'\Pi A \\ G'\Pi A
\end{bmatrix}
\]
or equivalently, with $b = \begin{bmatrix} B'\Pi A \\ G'\Pi A \end{bmatrix}$
    \begin{equation}
        \begin{bmatrix}
K \\ J
\end{bmatrix} =-M(\lambda)^{\dagger}b+\mathcal{N}(M(\lambda)) \label{JK2}
    \end{equation}
For any $v \in \mathcal{N}(M(\lambda))$, we have $M(\lambda)v = 0$, which gives $v'M(\lambda)v = 0$ and $(M(\lambda)^{\dagger}b)'M(\lambda)v = b'M(\lambda)^{\dagger}M(\lambda)v = 0$. Therefore, when substituting $\begin{bmatrix} K \\ J \end{bmatrix} = -M(\lambda)^{\dagger}b + v$ into the quadratic form $\begin{bmatrix} K' & J' \end{bmatrix} M(\lambda) \begin{bmatrix} K \\ J\end{bmatrix}$, all terms involving $v$ vanish. Thus, the following expression
\begin{equation}
        \Pi(\lambda) =Q+A'\Pi A- \begin{bmatrix} K' & J' \end{bmatrix}
     M(\lambda)
    \begin{bmatrix} K \\ J\end{bmatrix} \label{kj2}
\end{equation}
is equivalent to \eqref{recjk}. Expanding \eqref{kj2}
\begin{align*}
    \Pi(\lambda) = & Q+A'\Pi A-K'B'\Pi BK -K'RK-K'B'\Pi GJ-J'G'\Pi BK\\&-J'(G'\Pi G-\lambda I)J
\end{align*}
Consider
\[
B'\Pi GJ = -B' \Pi A - (B'\Pi B + R)K
\]
and
\[
J = -(G' \Pi G-\lambda I)^{\dagger}G' \Pi (A+BK)+\mathcal{N}(G' \Pi G-\lambda I)
\]
For any $q \in \mathcal{N}(G' \Pi G-\lambda I)$, we have $(G' \Pi G-\lambda I)q = 0$, giving $q'(G' \Pi G-\lambda I)q = 0$ and $((G' \Pi G-\lambda I)^{\dagger}G' \Pi (A+BK))'(G' \Pi G-\lambda I)q = 0$. Therefore, all terms involving $q$ vanish in the quadratic form $J'(G'\Pi G-\lambda I)J$.
Thus, substituting $B'\Pi GJ$ and $J$ in $\Pi(\lambda) = Q+A'\Pi A-K'B'\Pi BK -K'RK-K'B'\Pi GJ-J'G'\Pi BK-J'(G'\Pi G-\lambda I)J$ we obtain
\begin{align*}
\Pi(\lambda) =& Q+K'RK+(A+BK)'\Pi(A+BK)\\&-(A+BK)'\Pi G(G'\Pi G-\lambda I)^{\dagger}G' \Pi(A+BK)
\end{align*}
which is \eqref{xlrec} with $\bar{A} = A+BK$ and $\bar{Q} = Q + K'RK$.
\end{proof}

The following classical result relates positive definiteness of a block matrix to its Schur complement. It is used in deriving the LMI in \Cref{lemma:lmi,lemma:lmisd}.

\begin{proposition}[Schur complement and positive definiteness]
\label{prop:schur-complement}
Let $M$ be a symmetric matrix of the form
\[
M = \begin{bmatrix}
A & B\\
B' & C
\end{bmatrix}
\]
If $C$ is invertible, then the following properties hold:
\begin{enumerate}
\item $M \succ 0$ if and only if $C \succ 0$ and $A - BC^{-1}B' \succ 0$.
\item $M \succeq 0$ if and only if $C \succ 0$ and $A - BC^{-1}B' \succeq 0$.
\end{enumerate}
\end{proposition}

\begin{proof}
Both claims are the standard Schur complement condition for definiteness of a symmetric block matrix with invertible $C$, given in \citet[App.~A.5.5]{boyd:vandenberghe:2004}. In item~2, $C$ invertible together with $C \succeq 0$ implies $C \succ 0$.
\end{proof}

\begin{remark}
An analogous result holds using the Schur complement of $A$ instead of $C$. When $A$ is invertible, $M \succ 0$ if and only if $A \succ 0$ and $C - B'A^{-1}B \succ 0$.
\end{remark}

The following lemma establishes the equivalence between the Riccati inequality and its LMI representation under strict inequality. This result underlies the tractable optimization of the infinite horizon SiDAR in \Cref{prop:lmigen}.

\begin{lemma}[LMI equivalence to Riccati inequality I]
   \label{lemma:lmi}
   Assume $G'\Pi G - \lambda I \prec0$, $\Pi\succ0$, $\Pi - \bar{Q} - \bar{A}'\Pi\bar{A} \succ0$, and $\lambda>0$. The following linear matrix inequality (LMI)
\begin{equation*}
\begin{gathered}
   \begin{bmatrix}
      P & (AP - BF)' & 0 & P\hat{Q}' - F'\hat{R}'\\
      AP - BF & P & G & 0 \\
      0 & G' & \lambda I & 0\\
      (P\hat{Q}' - F'\hat{R}')'  & 0 & 0 & I
   \end{bmatrix}
   \succ 0 \\
   K = -F P^{-1} \quad P = \Pi^{-1}
\end{gathered}
\end{equation*}
   holds if and only if
   \begin{equation*}
   \Pi - \bar{Q} - \bar{A}'\Pi\bar{A} - \bar{A}'\Pi G (\lambda I - G'\Pi G)^{-1}G'\Pi\bar{A} \succ 0
   \end{equation*}
   where $\bar{A} = A+BK$, $\bar{Q} = Q + K'RK$, $\hat{Q}'=\begin{bmatrix}
       Q^{1/2}&
       0
     \end{bmatrix}$, and $
     \hat{R}'
     =
     \begin{bmatrix}
       0&
       R^{1/2}
     \end{bmatrix}
   $.
\end{lemma}

\begin{proof}
Define $K \eqbyd -F P^{-1}$, $\bar{A} \eqbyd A + BK$, $C \eqbyd \begin{bmatrix} Q^{1/2}\\ R^{1/2}K\end{bmatrix}$, and $\bar{Q} \eqbyd C'C = Q + K'RK$, where
\[
\hat{Q}' \eqbyd \begin{bmatrix} Q^{1/2} & 0 \end{bmatrix} \qquad \hat{R}' \eqbyd \begin{bmatrix} 0 & R^{1/2} \end{bmatrix}
\]
Substituting these definitions, the LMI becomes
\[
\begin{bmatrix}
P & P \bar{A}' & 0 & P C' \\
\bar{A} P & P & G & 0 \\
0 & G' & \lambda I & 0 \\
C P & 0 & 0 & I
\end{bmatrix} \succ 0
\]
Define $\Pi \eqbyd P^{-1} \succ 0$ and apply the congruence transformation $\mathcal{T} \eqbyd \operatorname{diag}(\Pi, I, I, I)$. Since congruence transformations preserve positive definiteness bidirectionally, we obtain
\[
\begin{bmatrix}
\Pi & \bar{A}' & 0 & C' \\
\bar{A} & \Pi^{-1} & G & 0 \\
0 & G' & \lambda I & 0 \\
C & 0 & 0 & I
\end{bmatrix} \succ 0
\]
Applying a permutation (reordering rows and columns), which preserves positive definiteness bidirectionally, yields
\[
\begin{bmatrix}
\Pi^{-1} & \bar{A} & G & 0 \\
\bar{A}' & \Pi & 0 & C' \\
G' & 0 & \lambda I & 0 \\
0 & C & 0 & I
\end{bmatrix} \succ 0
\]
Since $\Pi^{-1} \succ 0$ and by \cref{prop:schur-complement} item 1, this holds if and only if the Schur complement with respect to the $(1,1)$ block is positive definite
\[
\begin{bmatrix}
\Pi - \bar{A}' \Pi \bar{A} & -\bar{A}' \Pi G & C' \\
-G' \Pi \bar{A} & \lambda I - G' \Pi G & 0 \\
C & 0 & I
\end{bmatrix} \succ 0
\]
Since $I \succ 0$ and by \cref{prop:schur-complement} item 1, this holds if and only if the Schur complement with respect to the $(3,3)$ block is positive definite
\[
\begin{bmatrix}
\Pi - \bar{Q} - \bar{A}'\Pi\bar{A} & -\bar{A}' \Pi G \\
-G'\Pi \bar{A} & \lambda I - G'\Pi G
\end{bmatrix} \succ 0
\]
where we used $\bar{Q} = C'C$. Since $\lambda I - G'\Pi G \succ 0$ by assumption and by \cref{prop:schur-complement} item 1, this holds if and only if the Schur complement with respect to the $(2,2)$ block is positive definite
\[
(\Pi - \bar{Q} - \bar{A}'\Pi\bar{A}) - (-\bar{A}'\Pi G)(\lambda I - G'\Pi G)^{-1}(-\bar{A}'\Pi G)' \succ 0
\]
which is well-defined since $\lambda I-G'\Pi G\succ0$. Simplifying yields the equivalent form
\[
\Pi - \bar{Q} - \bar{A}'\Pi\bar{A} - \bar{A}'\Pi G(\lambda I - G'\Pi G)^{-1}G'\Pi\bar{A} \succ 0
\]
where the chain of equivalences is reversible, completing the proof.
\end{proof}

The following lemma extends the LMI equivalence to positive semidefiniteness, which is required for the steady-state characterization in \Cref{sec:ss}.

\begin{lemma}[LMI equivalence to Riccati inequality II]
   \label{lemma:lmisd}
   Assume $\Pi \succ0$ and $\lambda>0$. The following linear matrix inequality (LMI)
\begin{equation*}
\begin{gathered}
   \begin{bmatrix}
      P & (AP - BF)' & 0 & P\hat{Q}' - F'\hat{R}'\\
      AP - BF & P & G & 0 \\
      0 & G' & \lambda I & 0\\
      (P\hat{Q}' - F'\hat{R}')'  & 0 & 0 & I
   \end{bmatrix}
   \succeq 0 \\
   K = -F P^{-1} \quad P = \Pi^{-1}
\end{gathered}
\end{equation*}
   holds if and only if
   \begin{equation*}
   \begin{bmatrix}
       \Pi -\bar{Q}- \bar{A}'\Pi \bar{A} & -\bar{A}' \Pi G \\
       -G'\Pi \bar{A} & \lambda I-G'\Pi G
   \end{bmatrix} \succeq 0
   \end{equation*}
   where $\bar{A} = A+BK$, $\bar{Q} = Q + K'RK$, $\hat{Q}'=\begin{bmatrix}
       Q^{1/2}&
       0
     \end{bmatrix}$, and $
     \hat{R}'
     =
     \begin{bmatrix}
       0&
       R^{1/2}
     \end{bmatrix}
   $.
\end{lemma}

\begin{proof}
The proof follows that of \cref{lemma:lmi} with $\succeq$ replacing $\succ$ throughout and \cref{prop:schur-complement} item~2 in place of item~1. Substituting the definitions of $K$, $\bar{A}$, $C$, and $\bar{Q}$, the congruence transformation $\mathcal{T} \eqbyd \operatorname{diag}(\Pi, I, I, I)$ with $\Pi \eqbyd P^{-1} \succ 0$ and a permutation reduce the LMI to the same block form, and Schur complements with respect to the $(1,1)$ and $(3,3)$ blocks give the stated $2 \times 2$ matrix inequality. Under the semidefinite hypotheses, $\lambda I - G'\Pi G$ need not be invertible, so the final Schur reduction taken in the proof of \cref{lemma:lmi} does not apply, and the $2 \times 2$ inequality is the conclusion.
\end{proof}

\subsection{\textbf{Proofs of Main Results}}
\label{app:proofs}

\begin{proof}[Proof of \cref{prop:lmigen}]
We establish that optimization \eqref{inflmi} implies \eqref{eq:ss-problem} by invoking strong duality and complementary slackness for the convex 4$\times$4 semidefinite program. Following \citep{balakrishnan:vandenberghe:2003}, consider the following primal and dual optimization problems.

\textbf{Primal SDP.} \newline
For all $x_0 \in \bbR^n$ and $\alpha >0$, let primal variables be $(P,F,\lambda,\chi)$ with $P \in \mathbb{S}^n$, $F \in \mathbb{R}^{m \times n}$, $\lambda,\chi \in \mathbb{R}$. Define
\begin{equation}
\begin{split}
S(P,F,\lambda) &\eqbyd \begin{bmatrix}
       P & (AP - BF)' & 0 & P\hat{Q}' - F'\hat{R}'\\
       AP - BF & P & G & 0 \\
       0 & G' & \lambda I & 0\\
       (P\hat{Q}' - F'\hat{R}')'  & 0 & 0 & I
    \end{bmatrix} \\
T(P,\chi) &\eqbyd \begin{bmatrix}
       P & \frac{x_0}{\sqrt{\alpha}}\\
       \frac{x_0'}{\sqrt{\alpha}} & \chi 
    \end{bmatrix}
\end{split}
\label{eq:primal-blocks}
\end{equation}
where $\hat{Q}' = \begin{bmatrix} Q^{1/2} & 0 \end{bmatrix}$ and $\hat{R}' = \begin{bmatrix} 0 & R^{1/2} \end{bmatrix}$. The primal problem is
\begin{equation}
\min_{P,F,\lambda,\chi} \; \lambda /2 + \chi/2 \quad \text{s.t.} \quad S(P,F,\lambda) \succeq 0 \quad T(P,\chi) \succeq 0
\label{eq:primal-sdp}
\end{equation}
which is equivalent to \eqref{inflmi}.

Using the epigraph formulation, by Schur complement (see \citet[App.~B.1 \& App.~A.5.5]{boyd:vandenberghe:2004}) the quadratic term $(1/2)(x_0/\sqrt{\alpha})'\Pi (x_0/\sqrt{\alpha}) = (1/2)(x_0/\sqrt{\alpha})'P^{-1}(x_0/\sqrt{\alpha})$ is implied by $\begin{bmatrix} P & \frac{x_0}{\sqrt{\alpha}} \\ \frac{x_0'}{\sqrt{\alpha}} & \chi \end{bmatrix} \succeq 0$. Thus, optimization \eqref{eq:primal-sdp} is equivalently expressed by
\begin{equation}
\begin{split}
\min_{\lambda,\Pi,F} &(1/2)(\frac{x_0}{\sqrt{\alpha}})'\Pi (\frac{x_0}{\sqrt{\alpha}}) + \lambda/2 \\
\text{s.t.} \quad &S(P,F,\lambda) \succeq 0 \\
&K = - F P^{-1} \quad P = \Pi^{-1} \succ 0
\end{split}
\label{eq:primal-sdp2}
\end{equation}

\textbf{Dual SDP.} \newline
Define dual variables
\[
Y = \begin{bmatrix}
Y_{11} & Y_{12} & Y_{13} & Y_{14}\\
Y_{12}' & Y_{22} & Y_{23} & Y_{24}\\
Y_{13}' & Y_{23}' & Y_{33} & Y_{34}\\
Y_{14}' & Y_{24}' & Y_{34}' & Y_{44}
\end{bmatrix} \succeq 0 \qquad
V = \begin{bmatrix}
V_{11} & v\\
v' & \nu
\end{bmatrix} \succeq 0
\]
Using standard SDP duality \citep{balakrishnan:vandenberghe:2003}, the dual of the primal LMI epigraph problem~\eqref{eq:primal-sdp} is
\begin{equation}
\max_{Y,V} \; -\operatorname{Tr}(Y_{44}) - 2\operatorname{Tr}(G'Y_{23})- 2(\frac{x_0}{\sqrt{\alpha}})'v
\label{eq:dual-sdp}
\end{equation}
subject to $Y \succeq 0$, $V \succeq 0$, and the dual equalities
\begin{align*}
Y_{11} + A'Y_{12} + Y_{12}'A + Y_{22} + Y_{14}\hat{Q} + \hat{Q}'Y_{14}' + V_{11} &= 0\\
B'Y_{12} + \hat{R}'Y_{14}' &= 0 \\
\operatorname{Tr}(Y_{33}) &= \tfrac{1}{2} \\
\nu &= \tfrac{1}{2}
\end{align*}

\textbf{Strict primal feasibility.} \newline
To invoke strong duality, we establish existence of a strictly feasible primal point. Since $(A,B)$ is stabilizable by Assumption~\ref{asst1}, there exists $K$ such that $\bar{A} \eqbyd A + BK$ is Schur stable with $\rho(\bar{A}) < 1$. Define $C \eqbyd \begin{bmatrix} Q^{1/2}\\ R^{1/2}K\end{bmatrix}$ so that $\bar{Q} \eqbyd C'C = Q + K'RK$. For any $X \succ 0$ (e.g., $X = I$), the discrete Lyapunov equation $\Pi - \bar{A}'\Pi\bar{A} = \bar{Q} + X$ has unique solution $\Pi \succ 0$ since $\rho(\bar{A}) < 1$. By construction, $X = \Pi - \bar{Q} - \bar{A}'\Pi\bar{A} \succ 0$. Define $H(\lambda)\eqbyd \bar{A}'\Pi G(\lambda I-G'\Pi G)^{-1}G'\Pi \bar{A}$ for $\lambda>\norm{G'\Pi G}$. Then $H(\lambda)\succeq 0$ since $\lambda I-G'\Pi G\succ0$, $H(\lambda)$ is continuous in $\lambda$ on $(\norm{G'\Pi G},\infty)$, and $H(\lambda)\to0$ as $\lambda\to\infty$, so we choose $\lambda$ sufficiently large that $H(\lambda)\prec X$. This choice guarantees
\[
\Pi - \bar{Q} - \bar{A}'\Pi\bar{A} - \bar{A}'\Pi G (\lambda I - G'\Pi G)^{-1}G'\Pi\bar{A} \succ 0
\]
since $X\succ 0$ and $H(\lambda)\prec X$ imply $X-H(\lambda)\succ0$. By \cref{lemma:lmi}, with $P \eqbyd \Pi^{-1}$ and $F \eqbyd -KP$, this inequality holds if and only if $S(P,F,\lambda) \succ 0$. For the second constraint, choose $\chi \eqbyd (x_0/\sqrt{\alpha})'\Pi(x_0/\sqrt{\alpha}) + \varepsilon$ for any $\varepsilon > 0$. Since $P = \Pi^{-1} \succ 0$, the Schur complement of $T(P,\chi)$ yields
\begin{align*}
\chi - (x_0/\sqrt{\alpha})'P^{-1}(x_0/\sqrt{\alpha}) &= \chi - (x_0/\sqrt{\alpha})'\Pi(x_0/\sqrt{\alpha}) \\
&= \varepsilon > 0
\end{align*}
hence $T(P,\chi) \succ 0$. Therefore, $(P,F,\lambda,\chi)$ is a strictly feasible point, establishing that the primal problem \eqref{eq:primal-sdp2}, and thus problem \eqref{eq:primal-sdp}, is strictly feasible.

\textbf{Strong duality and complementary slackness.} \newline
Since we established existence of a strictly feasible primal point in the previous paragraph, Theorem~4 of \citep{balakrishnan:vandenberghe:2003} guarantees that strong duality holds. Therefore, at optimality $(P, F, \lambda, \chi)$ and $(Y, V)$ satisfy the complementary slackness conditions
\begin{equation}
\langle S(P,F,\lambda^*), Y \rangle = 0 \qquad \langle T(P,\chi), V \rangle = 0
\label{eq:cs}
\end{equation}
Define $\Pi \eqbyd (P)^{-1}$, $K \eqbyd -F(P)^{-1}$, $\bar{A} \eqbyd A + BK$, and $\bar{Q} \eqbyd Q + K'RK = C'C$ with $C \eqbyd \begin{bmatrix} Q^{1/2}\\ R^{1/2}K\end{bmatrix}$, so that $C' = \begin{bmatrix} Q^{1/2} & K'R^{1/2}\end{bmatrix}$ has compatible block dimensions for arbitrary $n$, $m$, and $q$. Apply congruence transformation $\mathcal{T} = \operatorname{diag}(\Pi, I, I, I)$ to $S(P,F,\lambda)$ to obtain
\[
\Sigma \eqbyd \mathcal{T}'S(P,F,\lambda)\mathcal{T} = \begin{bmatrix}
\Pi & \bar{A}' & 0 & C'\\
\bar{A} & \Pi^{-1} & G & 0\\
0 & G' & \lambda I & 0\\
C & 0 & 0 & I
\end{bmatrix} \succeq 0
\]
and $\Upsilon \eqbyd \mathcal{T}^{-1}Y^*(\mathcal{T}^{-1})' \succeq 0$. Since $\langle \Sigma, \Upsilon \rangle = \langle S(P,F,\lambda), Y \rangle = 0$, we have $\operatorname{Tr}(\Sigma\Upsilon) = 0$.


Define $W \eqbyd \Upsilon^{1/2}$ so that $\Upsilon = WW'$ (principal square root). Since $\Sigma \succeq 0$, $\Upsilon \succeq 0$, and $\operatorname{Tr}(\Sigma\Upsilon) = 0$, we apply the following result: if $A, B \succeq 0$ and $\operatorname{Tr}(AB) = 0$, then with $W \eqbyd B^{1/2}$ we have
\begin{align*}
\operatorname{Tr}(AB) &= \operatorname{Tr}(A^{1/2}WW'A^{1/2}) \\
&= \operatorname{Tr}((A^{1/2}W)(A^{1/2}W)') \\
&= \normf{A^{1/2}W}_F^2 = 0
\end{align*}
which implies $A^{1/2}W = 0$ and hence $AW = 0$. Applying this with $A = \Sigma$ and $B = \Upsilon$ yields
\begin{equation}
\Sigma W = 0
\label{eq:sigma-w-zero}
\end{equation}

Partition $W = [W_1; W_2; W_3; W_4]$ conformably with $\Sigma$. The four block rows of \eqref{eq:sigma-w-zero} are
\begin{align}
\Pi W_1 + \bar{A}'W_2 + C'W_4 &= 0 \label{eq:block-a}\\
\bar{A}W_1 + \Pi^{-1}W_2 + GW_3 &= 0 \label{eq:block-b}\\
G'W_2 + \lambda W_3 &= 0 \label{eq:block-c}\\
CW_1 + W_4 &= 0 \label{eq:block-d}
\end{align}
From \eqref{eq:block-d}, $W_4 = -CW_1$. From \eqref{eq:block-b}, since $\Pi^{-1} \succ 0$, we have
\[
W_2 = -\Pi(\bar{A}W_1 + GW_3)
\]
Substituting $W_2$ and $W_4$ into \eqref{eq:block-a} gives
\[
(\Pi - \bar{Q} - \bar{A}'\Pi\bar{A})W_1 - \bar{A}'\Pi GW_3 = 0
\]
and substituting $W_2$ into \eqref{eq:block-c} gives
\[
-G'\Pi\bar{A}W_1 + (\lambda I - G'\Pi G)W_3 = 0
\]
Stack these two relations as
\[
\begin{bmatrix}
\Pi - \bar{Q} - \bar{A}'\Pi\bar{A} & -\bar{A}'\Pi G\\
-G'\Pi\bar{A} & \lambda I - G'\Pi G
\end{bmatrix}
\begin{bmatrix} W_1 \\ W_3 \end{bmatrix} = 0
\]
Suppose $\Upsilon_{11} = W_1W_1' \succ 0$. Then, $W_1$ has a right inverse 
$W_1^{\dagger} \eqbyd W_1'(W_1W_1')^{-1}$. Define
\[
\tilde{J}' \eqbyd W_3 W_1^{\dagger}
\]
Multiplying the last equation on the right by $W_1^{\dagger}$ yields
\[
\begin{bmatrix}
\Pi - \bar{Q} - \bar{A}'\Pi\bar{A} & -\bar{A}'\Pi G\\
-G'\Pi\bar{A} & \lambda I - G'\Pi G
\end{bmatrix}
\begin{bmatrix} I \\ \tilde{J}' \end{bmatrix} = 0
\]
    or equivalently $\tilde{J}' = (\lambda I-G'\Pi G)^{\dagger}G'\Pi \bar{A}+\mathcal{N}(\lambda I-G'\Pi G)$, and
    \(\Pi -C'C- \bar{A}'\Pi \bar{A} -\bar{A}' \Pi G\tilde{J}'=0\). After simplification, we obtain
    \[
    \Pi = \bar{Q} + \bar{A}' \Pi\bar{A} +\tilde{J}(\lambda I-G'\Pi G)\tilde{J}'
    \]
    which can be rewritten as
    \begin{equation}
\Pi = \bar{Q} + \bar{A}'\Pi\bar{A} - \bar{A}'\Pi G(G'\Pi G - \lambda I)^{\dagger}G'\Pi\bar{A}
\label{eq:riccati-equality}
\end{equation}
By \cref{receq}, \eqref{eq:riccati-equality} can be rewritten as $g(\lambda, \Pi) = 0$ from \eqref{eq:ss-problem:constraints},
where $\bar{A} = A+BK$ and $\bar{Q} = Q + K'RK$ and $K$ satisfies
    \begin{equation*}
  \begin{bmatrix}
        B'\Pi B+R & B'\Pi G \\
        G'\Pi B & G'\Pi G-\lambda I
    \end{bmatrix}\begin{bmatrix}
K \\ J
\end{bmatrix} =-\begin{bmatrix}
B'\Pi A \\ G'\Pi A
\end{bmatrix}
\end{equation*}
Note that since $M(\lambda,\Pi)$ is invertible for $\lambda\geq \norm{G'\Pi G}$ under Assumption 2, the Moore–Penrose pseudoinverse in \cref{receq} reduces to the standard inverse.
From \cref{lemma:lmisd}, the inequality $S(P,F,\lambda) \succeq 0$ implies $\lambda I - G'\Pi G \succeq 0$, and thus $\lambda \geq \norm{G'\Pi G}$. Therefore, $(\lambda, \Pi)$ satisfies the constraints of the steady-state problem \eqref{eq:ss-problem}.

\textbf{Conclusion.} \newline
The complementary slackness analysis establishes that any optimal solution $(\lambda^*, \chi^*, P^*, F^*)$ to  \eqref{eq:primal-sdp2}, and thus \eqref{inflmi}, yields $(\lambda^*, \Pi)$ with $\Pi = (P^*)^{-1}$ satisfying $g(\lambda^*, \Pi) = 0$, $\lambda^* \geq \norm{G'\Pi G}$. Therefore, optimization \eqref{eq:primal-sdp2}, and thus \eqref{inflmi}, implies the steady-state problem \eqref{eq:ss-problem}.

\end{proof}

\begin{proof}[Proof of \cref{prop:lmi-existence}]
We establish existence by showing that every feasible level set is compact, then applying the Weierstrass theorem. Let $f(P,F,\lambda,\chi) \eqbyd \lambda /2 + \chi/2$ denote the objective function.

From the $(3,3)$ block of $S(P,F,\lambda) \succeq 0$ from \eqref{eq:primal-blocks}, we have $\lambda I \succeq 0$, hence $\lambda \geq 0$. From the Schur complement of $T(P,\chi) \succeq 0$, we have $\chi P - (x_0 x_0')/\alpha \succeq 0$, which implies $\chi \geq 0$. Therefore, on any level set $\{f \leq \beta\}$, we have $0 \leq \lambda, \chi \leq 2\beta$.

Since the $(4,4)$ block of $S(P,F,\lambda)$ is $I$, the Schur complement with respect to this block yields
\[
\begin{bmatrix}
P & (AP - BF)' & 0\\
AP - BF & P & G\\
0 & G' & \lambda I
\end{bmatrix}
\succeq
\begin{bmatrix}
P\hat Q' - F'\hat R'\\[2pt]
0\\[2pt]
0
\end{bmatrix}
\begin{bmatrix}
P\hat Q' - F'\hat R'\\[2pt]
0\\[2pt]
0
\end{bmatrix}'
\]
The $(1,1)$ block of this inequality gives
\begin{equation}
P \succeq (P\hat{Q}' - F'\hat{R}')(P\hat{Q}' - F'\hat{R}')'
\label{eq:schur-11}
\end{equation}
Since $\hat{Q}' = \begin{bmatrix} Q^{1/2} & 0 \end{bmatrix}$ and $\hat{R}' = \begin{bmatrix} 0 & R^{1/2} \end{bmatrix}$, we have
\[
P\hat{Q}' - F'\hat{R}' = \begin{bmatrix} PQ^{1/2} & -F'R^{1/2} \end{bmatrix}
\]
Therefore
\begin{align*}
(P\hat{Q}' - F'\hat{R}')(P\hat{Q}' - F'\hat{R}')' &= PQ^{1/2}Q^{1/2}P + F'R^{1/2}R^{1/2}F \\
&= PQP + F'RF
\end{align*}
Substituting into \eqref{eq:schur-11} yields
\begin{equation}
P \succeq PQP + F'RF
\label{eq:schur-pf}
\end{equation}
Define $M \eqbyd Q^{1/2}PQ^{1/2} \succeq 0$. From \eqref{eq:schur-pf}, left- and right-multiplying by $Q^{1/2}$ gives
\begin{align*}
M &\succeq Q^{1/2}(PQP)Q^{1/2} + Q^{1/2}F'RFQ^{1/2} \\
&\succeq Q^{1/2}(PQP)Q^{1/2} = M^2
\end{align*}
Let $\mu$ be any eigenvalue of $M$ with corresponding eigenvector $v$. Then $Mv = \mu v$ and $M^2v = \mu^2 v$. Since $M \succeq M^2$, the quadratic form yields
\begin{align*}
0 &\leq v'(M - M^2)v = v'Mv - v'M^2v \\
&= \mu \norm{v}^2 - \mu^2 \norm{v}^2 = \mu(1-\mu)\norm{v}^2
\end{align*}
Since $M \succeq 0$, we have $\mu \geq 0$. Therefore $\mu(1-\mu) \geq 0$ implies $0 \leq \mu \leq 1$. Hence $M \preceq I$, which by the definition of $M$ gives
\begin{equation}
0 \preceq P \preceq Q^{-1}
\label{eq:P-bound}
\end{equation}
where the inverse exists by Assumption~\ref{asst4} ($Q \succ 0$). Thus $\norm{P} \leq \normf{Q^{-1}}$ on the feasible set.
From \eqref{eq:schur-pf}, we have $F'RF \preceq P - PQP \preceq P \preceq Q^{-1}$. Therefore
\[
\norm{F'RF} \leq \normf{Q^{-1}}
\]
Since $R \succ 0$ by Assumption~\ref{asst1}, the square root $R^{1/2}$ exists and
\[
\normf{R^{1/2}F}^2 = \normf{(R^{1/2}F)'(R^{1/2}F)} = \norm{F'RF} \leq \normf{Q^{-1}}
\]
Hence
\begin{equation}
\norm{F} \leq \normf{R^{-1/2}} \sqrt{\norm{Q^{-1}}}
\label{eq:F-bound}
\end{equation}

Consider the level set
\begin{align*}
\mathcal{L}_{\beta} \eqbyd \{(P,F,\lambda,\chi) : &S(P,F,\lambda) \succeq 0, \\
&T(P,\chi) \succeq 0, \\
&f(P,F,\lambda,\chi) \leq \beta\}
\end{align*}
From the bounds established above, any $(P,F,\lambda,\chi) \in \mathcal{L}_{\beta}$ satisfies
\[
0 \preceq P \preceq Q^{-1} \quad \norm{F} \leq c \quad 0 \leq \lambda \leq 2\beta \quad 0 \leq \chi \leq 2\beta
\]
where $c \eqbyd \normf{R^{-1/2}} \sqrt{\norm{Q^{-1}}}$. Therefore $\mathcal{L}_{\beta}$ is contained in a compact subset of $\mathbb{S}^n \times \mathbb{R}^{m \times n} \times \mathbb{R}^2$. Since the LMI constraints $S(P,F,\lambda) \succeq 0$ and $T(P,\chi) \succeq 0$ define closed sets (positive semidefinite matrices form a closed cone) and the objective $f$ is continuous, the level set $\mathcal{L}_{\beta}$ is closed and bounded, hence compact.

Assuming the feasible set is nonempty (which is guaranteed by $(A,B)$ stabilizable, as shown in the proof of \cref{prop:lmigen}), the infimum of $f$ over the feasible set is finite (bounded below by $0$). By compactness of level sets and continuity of $f$, the Weierstrass theorem ensures that the infimum is attained at some $(P^*,F^*,\lambda^*,\chi^*)$ in the feasible set. Therefore an optimal solution exists.
\end{proof}

\begin{proof}[Proof of \cref{infl}]

To simplify notation, we reverse the time indexing used in \cref{prop:ndsignal}.
Set
\[
P_0\eqbyd\Pi_N=P_f\qquad P_k\eqbyd\Pi_{N-k}\qquad(k\ge0)
\]
and, for any scalar $\gamma$, the following forward recursion at step $k$
mirrors the backward step at index $N-k$ from \eqref{signalrec1}
\begin{equation}
\begin{split}
P_{k+1}(\gamma) &= Q+A'P_kA
-A'P_k\!\begin{bmatrix}B&G\end{bmatrix} \\
&\quad \begin{bmatrix}
B'P_kB+R & B'P_kG\\
G'P_kB   & G'P_kG-\gamma I
\end{bmatrix}^{-1}
\!\begin{bmatrix}B'\\G'\end{bmatrix}P_kA
\end{split}
\label{reverse}
\end{equation}

\textit{Notation for multipliers:} We adopt the following conventions. The symbol $\gamma$ denotes a generic scalar multiplier (optimization variable) in $P_{k+1}(\gamma)$ and is identified with $\lambda$ under the backward indexing. Define the feasibility lower bound (stagewise) as
\begin{gather*}
\gamma_0\eqbyd\norm{G'P_0G}=\norm{G'P_fG}\\
\gamma_{k+1}\eqbyd
\begin{cases}
\displaystyle\min_{\gamma\ge\gamma_k}\bigl\{\gamma\,:\,\gamma=\norm{G'P_k(\gamma)G}\bigr\}
\quad &\text{if } \norm{G'P_k(\gamma_k)G}>\gamma_k\\[4pt]
\gamma_k \quad&\text{if } \norm{G'P_k(\gamma_k)G}\le\gamma_k
\end{cases}
\end{gather*}
so $\gamma_{k+1}\ge\gamma_k$. For fixed $x_0$ and $\alpha$, the stage-$k$ objective is
\[
L_k(\gamma)\eqbyd\frac{1}{2}\left(\frac{x_0}{\sqrt{\alpha}}\right)'\!P_k(\gamma)\left(\frac{x_0}{\sqrt{\alpha}}\right)+\frac{\gamma}{2}\qquad\gamma\in[\gamma_{k-1},\infty)
\]
Let $\gamma^\star_k$ denote the interior stationary point of $L_k(\cdot)$, satisfying $\frac{dL_k}{d\gamma}(\gamma^\star_k)=0$ if it exists (it need not lie in $[\gamma_{k-1},\infty)$). The stage-$k$ optimizer is
\[
\gamma^*(k)\eqbyd\arg\min_{\gamma\ge\gamma_{k-1}}L_k(\gamma)=\max\{\gamma_{k-1},\gamma^\star_k\}
\]
When the stage index is clear from context, $\gamma^*$ is shorthand for $\gamma^*(k)$. Define the limits
\[
\gamma_\infty\eqbyd\lim_{k\to\infty}\gamma_k\qquad
\gamma_\infty^*\eqbyd\lim_{k\to\infty}\gamma^*(k)\qquad
P_\infty(\gamma)\eqbyd\lim_{k\to\infty}P_k(\gamma)
\]
For horizon $N$, the backward time correspondence is $\gamma=\lambda$, $P_k(\gamma)=\Pi_{N-k}(\lambda)$, hence $\gamma^*(N)=\lambda^*(N)$ and $P_N(\gamma^*(N))=\Pi_0(\lambda^*(N))$. The steady-state solution is identified as $\overline{\lambda}\eqbyd\gamma_\infty^*$ and $\overline{\Pi}\eqbyd P_\infty(\gamma_\infty^*)$.
From \cref{prop:ndsignal}, $\norm{G'P_k(\gamma)G}$ is continuous and nonincreasing in $\gamma$ on $[\gamma_k,\infty)$, and recursion \eqref{reverse} is monotonic nondecreasing in $k$ and monotonic nonincreasing in $\gamma$ on $[\gamma_k,\infty)$.



The proof is structured as follow. First, we establish a uniform upper bound on $P_k(\gamma^*)$ and $\gamma^*$ in $k$. Then we prove the monotonicity and limit as $k\rightarrow\infty$ for both $\gamma^*(k)$, which is the optimal solution $\gamma^*$ at stage $k$, and $P_k(\gamma^*)$. Finally, we establish convergence of $(\gamma^*(k),P_k(\gamma^*(k)))$ as $k\rightarrow \infty$ to a minimal steady-state solution $(\overline{\lambda},\overline{\Pi})$.

\textbf{Upper bound on $P_k(\gamma^*)$ and $\gamma^*$.} \newline Since $(A,B)$ is stabilizable, there exists $K$ such that $A + BK$ is Schur stable. Let $\max\norm{\text{eig}(A+BK)}<\rho <1$ and define the possibly suboptimal control sequence $\tilde{\useq}$, where $u_i = Kx_i$ for $i \in \{0,1,\ldots,k-1\}$.
The closed-loop system under this control satisfies
\begin{equation*}
x_{i+1} = (A+BK)x_i + Gw_i
\end{equation*}
\begin{equation*}
x_i = (A+BK)^i x_0 + \sum_{j=0}^{i-1} (A+BK)^{i-1-j} Gw_j
\end{equation*}
for any disturbance sequence satisfying $\sum_{i=0}^{k-1} |w_i|^2 = \alpha$.
Since $\norm{(A+BK)^j} \leq c\rho^j$ for some $c > 0$, and using Cauchy-Schwarz inequality we obtain
\begin{align*}
\norm{x_i} &\leq c\rho^i \norm{x_0} + c\norm{G} \sum_{j=0}^{i-1} \rho^{i-1-j} \norm{w_j} \\
&\leq c\rho^i \norm{x_0} + c\norm{G} \left(\sum_{j=0}^{i-1} \rho^{2(i-1-j)}\right)^{1/2} \left(\sum_{j=0}^{i-1} \norm{w_j}^2\right)^{1/2} \\
&\leq c\rho^i \norm{x_0} + \frac{c\norm{G}\sqrt{\alpha}}{(1-\rho^2)^{1/2}}
\end{align*}
Define $\bar{Q} \eqbyd Q+K'RK$. The stage cost under this control satisfies
\begin{equation*}
\ell(x_i, u_i) = (1/2)x_i'Qx_i + (1/2)u_i'Ru_i = (1/2)x_i'\bar{Q}x_i \leq (1/2)\norm{\bar{Q}}\norm{x_i}^2
\end{equation*}
Express $x_i = y_i + z_i$ where $y_i \eqbyd (A+BK)^i x_0$ and $z_i \eqbyd \sum_{j=0}^{i-1} (A+BK)^{i-1-j} Gw_j$. Using $(a+b)^2 \leq 2a^2 + 2b^2$ we obtain
\begin{equation*}
\sum_{i=0}^{k-1} \norm{x_i}^2 \leq 2\sum_{i=0}^{k-1} \norm{y_i}^2 + 2\sum_{i=0}^{k-1} \norm{z_i}^2
\end{equation*}
For the homogeneous part, $\norm{y_i} \leq c\rho^i \norm{x_0}$ gives
\begin{equation*}
\sum_{i=0}^{k-1} \norm{y_i}^2 \leq c^2 \sum_{i=0}^{k-1} \rho^{2i} \norm{x_0}^2 = c^2 \frac{1-\rho^{2k}}{1-\rho^2} \norm{x_0}^2
\end{equation*}
For the forced part, define the sequence $\mathbf{a} \eqbyd (a_0, a_1, \dots)$ with $a_t \eqbyd c\norm{G}\rho^t$ for $t \geq 0$ and let $\mathbf{z} \eqbyd (z_0, z_1, \dots, z_{k-1})$. Since $\norm{(A+BK)^t} \leq c\rho^t$, we have $\norm{z_i} \leq \sum_{j=0}^{i-1} a_{i-1-j} \norm{w_j}$, which is a discrete convolution. Young's convolution inequality with truncated norms yields
\begin{equation*}
\smax{\mathbf{z}} \leq  \smax{\mathbf{a}}_1 \smax{\wseq}
\end{equation*}
where $ \smax{\mathbf{a}}_1=\sum_{t=0}^{k-1} a_t = c\norm{G} \sum_{t=0}^{k-1} \rho^t = c\norm{G}(1-\rho^k)/(1-\rho)$ and $\smax{\wseq}^2 = \sum_{j=0}^{k-1} \norm{w_j}^2 = \alpha$. Therefore
\begin{equation*}
\sum_{i=0}^{k-1} \norm{z_i}^2 = \smax{\mathbf{z}}^2 \leq \frac{c^2\norm{G}^2(1-\rho^k)^2}{(1-\rho)^2} \alpha
\end{equation*}
Combining both parts gives
\begin{align*}
V(x_0, \tilde{\useq}, \wseq)
      &= \sum_{i=0}^{k-1} \ell(x_i, u_i)
      \leq \frac{\norm{\bar{Q}}}{2} \sum_{i=0}^{k-1} \norm{x_i}^2 \\
      &\leq \frac{\norm{\bar{Q}}}{2} \left(2c^2 \frac{1-\rho^{2k}}{1-\rho^2} \norm{x_0}^2 + 2\frac{c^2\norm{G}^2(1-\rho^k)^2}{(1-\rho)^2} \alpha\right) \\
      &= \frac{\norm{\bar{Q}}c^2}{1-\rho^2} (1-\rho^{2k}) \norm{x_0}^2 + \frac{c^2\norm{G}^2\norm{\bar{Q}}}{(1-\rho)^2} (1-\rho^k)^2 \alpha
\end{align*}
Letting $k\to\infty$ and using $(1-\rho^{2k})\to 1$ and $(1-\rho^k)^2 \to 1$
\begin{align*}
V(x_0,\tilde{\useq},\wseq)
      &\le
         \frac{\norm{\bar{Q}}\,c^{2}}
              {1-\rho^{2}}\,\norm{x_0}^{2}+
         \frac{c^{2}\,\norm{G}^{2}\,\norm{\bar{Q}}}
              {(1-\rho)^2}\;\alpha
\end{align*}
Define
\begin{align*}
c_1 \eqbyd \frac{\norm{\bar{Q}}c^2}{1-\rho^2} \qquad c_2 \eqbyd \frac{c^{2}\,\norm{G}^{2}\,\norm{\bar{Q}}}
              {(1-\rho)^2}\;\alpha
\end{align*}
Since the optimal control achieves lower cost than the suboptimal control
\begin{equation*}
V^*(x_0) \leq \frac{V(x_0, \tilde{\useq}, \wseq)}{\alpha} \leq \frac{c_1}{\alpha}\norm{x_0}^2 + \frac{c_2}{\alpha}
\end{equation*}
From the optimality conditions in \cref{prop:ndsignal} the optimal cost is
\begin{equation*}
V^*(x_0) = (1/2)(\frac{x_0}{\sqrt{\alpha}})' P_k(\gamma^*) (\frac{x_0}{\sqrt{\alpha}}) + \gamma^*/2
\end{equation*}
where $P_k(\gamma)$ is generated by the forward recursion \eqref{reverse} and initial penalty $P_0 = P_f$.

\textbf{Limit of $\gamma^*(k)$.} \newline
Fix $x_0$. From \cref{prop:ndsignal}, $\gamma_k$ is monotonic nondecreasing in $k$, $\gamma_{k+1} \geq \gamma_k$.
At stage $k$, we minimize $L_k(\gamma) = (1/2)(x_0 / \sqrt{\alpha})' P_k(\gamma) (x_0 / \sqrt{\alpha}) + \gamma/2$ over $\gamma \in [\gamma_{k-1}, \infty)$. From Proposition 8 in in Mannini and Rawlings~\citep{mannini:rawlings:2026a}, the derivative of $L_k(\gamma)$ is nondecreasing on $[\gamma_{k-1},\infty)$ and given by
\[
\frac{dL_k}{d\gamma}(\gamma) = \frac{1}{2} - \frac{1}{2}\frac{\normf{\tilde{J}_k(\gamma)x_0}^2}{\alpha}
\]
where
\[
\tilde{J}_k(\gamma) \eqbyd \begin{bmatrix} J_0 \\ J_1 \Phi_{1} \\ \vdots \\ J_{k-1} \Phi_{k-1} \end{bmatrix}
\]
Since $\tilde{J}_{k+1}(\gamma)$ is obtained by stacking one additional block over $\tilde{J}_{k}(\gamma)$, we have $\tilde{J}_{k+1}'(\gamma)\tilde{J}_{k+1}(\gamma) \succeq \tilde{J}_{k}'(\gamma)\tilde{J}_{k}(\gamma)$, which implies
\[
\frac{dL_{k+1}}{d\gamma}(\gamma) \leq \frac{dL_{k}}{d\gamma}(\gamma) \quad \text{for all } \gamma \in [\gamma_{k}, \infty)
\]
Let $\gamma^\star_k$ be the unique interior root of $dL_k/d\gamma=0$, if it exists (it need not lie in $[\gamma_{k-1},\infty)$). The minimizer over the feasible set is
\[
\gamma^*(k) = \arg\min_{\gamma \geq \gamma_{k-1}} L_k(\gamma) = \max\{\gamma_{k-1}, \gamma^\star_k\}
\]
We then consider two cases.

\emph{Case A:} $\gamma_k^\star \le \gamma_k$. \quad
$\gamma^*(k)=\max\{\gamma_{k-1},\gamma_k^\star\}\le \gamma_k \le \gamma^*(k+1)$ (since stage $k+1$ minimizes over $[\gamma_k,\infty)$).

\emph{Case B:} $\gamma_k^\star \ge \gamma_k$. \quad
On $[\gamma_k,\infty)$, $dL_{k+1}/d\gamma\le dL_k/d\gamma$ and both are nondecreasing in $\gamma$, thus $\gamma_{k+1}^\star \ge \gamma_k^\star$. Hence
\begin{align*}
\gamma^*(k+1) &= \max\{\gamma_k,\gamma_{k+1}^\star\} \\
&\ge \max\{\gamma_k,\gamma_k^\star\} \\
&\ge \max\{\gamma_{k-1},\gamma_k^\star\} = \gamma^*(k)
\end{align*}
Therefore, from cases A and B we have that $\gamma^*(k)$ is monotonic nondecreasing in $k$.

Consider the previously obtained bound $V^*(x_0) \leq (c_1/\alpha)\norm{x_0}^2 + c_2/\alpha$. Since $V^*(x_0) = L_k(\gamma^*(k)) \geq \gamma^*(k)/2$, we obtain
\[
\gamma^*(k) \leq \frac{2c_1}{\alpha}\norm{x_0}^2 + \frac{2c_2}{\alpha} =: C(x_0)
\]
for each fixed $x_0$ and all $k$. Therefore, for each fixed $x_0$, $\gamma^*(k)$ is monotone nondecreasing and bounded, hence converges to $\gamma^*_{\infty}$ as $k \to \infty$
\[
\lim_{k\to\infty}\gamma^*(k) \;=\; \gamma^*_{\infty}
\]

\textbf{Limit of $P_k(\gamma^*(k))$.} \newline
We first establish uniform boundedness of $P_k(\gamma)$ for fixed $\gamma \geq \gamma_{k-1}$. For any fixed $\gamma \geq \gamma_{k-1}$, the Lagrangian value function satisfies
\begin{align*}
L_k(\gamma) &= \frac{1}{\alpha}\min_{u_0}\max_{w_0} \cdots \min_{u_{k-1}}\max_{w_{k-1}}\left( V(x_0,\useq,\wseq) - \frac{\gamma}{2}(|\wseq|^2-\alpha)\right) \\
&= \frac{1}{2}\left(\frac{x_0}{\sqrt{\alpha}}\right)'P_k(\gamma)\left(\frac{x_0}{\sqrt{\alpha}}\right) + \frac{\gamma}{2}
\end{align*}
Consider any unit vector $v$ with $\norm{v}=1$ and set $x_0=\sqrt{\alpha}\,v$. Then $\norm{x_0}^2=\alpha$ and $(x_0/\sqrt{\alpha})=v$, so we have
\[
L_k(\gamma) = \frac{1}{2}\,v'P_k(\gamma)v + \frac{\gamma}{2}
\]
Using the suboptimal stabilizing control $\tilde{\useq}$ with gain $K$ such that $A+BK$ is Schur stable, we obtain the upper bound
\begin{align*}
L_k(\gamma) &= \frac{1}{\alpha}\min_{u_0}\max_{w_0} \cdots \min_{u_{k-1}}\max_{w_{k-1}}\left( V(x_0,\useq,\wseq) - \frac{\gamma}{2}(|\wseq|^2-\alpha)\right) \\
&\leq \frac{1}{\alpha}\max_{w_0} \cdots \max_{w_{k-1}}\left( V(x_0,\tilde{\useq},\wseq) - \frac{\gamma}{2}(|\wseq|^2-\alpha)\right)
\end{align*}
Since $-(\gamma/2)|\wseq|^2 \leq 0$ for $\gamma \geq 0$, we have $-(\gamma/2)(|\wseq|^2-\alpha) \leq (\gamma/2)\alpha$, hence we obtain
\[
L_k(\gamma) \leq \frac{1}{\alpha}\max_{w_0} \cdots \max_{w_{k-1}} V(x_0,\tilde{\useq},\wseq) + \frac{\gamma}{2}
\]
From the stabilizing $K$ bound established earlier $$(1/\alpha)\max_{w_0} \cdots \max_{w_{k-1}} V(x_0,\tilde{\useq},\wseq) \leq (c_1/\alpha)\norm{x_0}^2 + c_2/\alpha$$ Substituting $\norm{x_0}^2 = \alpha$ yields
\[
\frac{1}{\alpha}\max_{w_0} \cdots \max_{w_{k-1}} V(x_0,\tilde{\useq},\wseq) \leq c_1 + \frac{c_2}{\alpha}
\]
Combining these inequalities gives
\[
\frac{1}{2}\,v'P_k(\gamma)v + \frac{\gamma}{2} = L_k(\gamma) \leq c_1 + \frac{c_2}{\alpha} + \frac{\gamma}{2}
\]
Canceling $\gamma/2$ on both sides yields $v'P_k(\gamma)v \leq 2c_1 + 2c_2/\alpha$. Since this holds for every unit vector $v$, we obtain $\norm{P_k(\gamma)} \leq 2c_1 + 2c_2/\alpha =: C_P$ for all $k$ and all $\gamma \geq \gamma_{k-1}$.

Since $\gamma^*(k) \geq \gamma_{k-1}$, monotonicity in $\gamma$ gives $P_k(\gamma^*(k)) \preceq P_k(\gamma_{k-1})$. Thus, we have
\[
\norm{P_k(\gamma^*(k))} \leq \norm{P_k(\gamma_{k-1})} \leq C_P
\]
Similarly, since $\gamma^*(k) \leq \gamma_\infty^*$, monotonicity in $\gamma$ gives $P_k(\gamma_\infty^*) \preceq P_k(\gamma^*(k))$. Thus, we have
\[
\norm{P_k(\gamma_\infty^*)} \leq \norm{P_k(\gamma^*(k))} \leq C_P
\]
hence from monotonicity in $k$ of $P_k(\gamma)$, we obtain $P_k(\gamma_\infty^*) \nearrow P_\infty(\gamma_\infty^*)$.

We now establish convergence $P_k(\gamma^*(k)) \to P_\infty(\gamma_\infty^*)$. Fix $x_0$ and $\gamma^- \in (\gamma_\infty, \gamma_\infty^*)$. Since $\gamma^*(k) \to \gamma_\infty^*$ from below, there exists $K$ such that for all $k \geq K$, the following inequalities hold
\[
\gamma^- \leq \gamma^*(k) \leq \gamma_\infty^*
\]
Since $\gamma \mapsto P_k(\gamma)$ is nonincreasing, for all $k \geq K$ we have the Loewner ordering
\[
P_k(\gamma_\infty^*) \preceq P_k(\gamma^*(k)) \preceq P_k(\gamma^-)
\]
For any $x \in \bbR^n$, this yields
\begin{equation}
x'P_k(\gamma_\infty^*)x \;\leq\; x'P_k(\gamma^*(k))x \;\leq\; x'P_k(\gamma^-)x
\label{eq:sandwich-quad}
\end{equation}
Since $\gamma^-, \gamma_\infty^* \geq \gamma_\infty$, the limits $P_\infty(\gamma^-)$ and $P_\infty(\gamma_\infty^*)$ are well-defined via monotone convergence $P_k(\gamma) \nearrow P_\infty(\gamma)$. Taking $k \to \infty$ in \eqref{eq:sandwich-quad} gives
\begin{align*}
x'P_\infty(\gamma_\infty^*)x &\leq \liminf_{k \to \infty} x'P_k(\gamma^*(k))x \\
&\leq \limsup_{k \to \infty} x'P_k(\gamma^*(k))x \\
&\leq x'P_\infty(\gamma^-)x
\end{align*}
We now take $\gamma^- \uparrow \gamma_\infty^*$. 
Therefore, we have
\[
\lim_{\gamma^- \uparrow \gamma_\infty^*} x'P_\infty(\gamma^-)x = x'P_\infty(\gamma_\infty^*)x
\]
Applying this limit to the bounds above yields, for every $x \in \bbR^n$, we obtain
\[
\liminf_{k \to \infty} x'P_k(\gamma^*(k))x = \limsup_{k \to \infty} x'P_k(\gamma^*(k))x = x'P_\infty(\gamma_\infty^*)x
\]
and hence
\begin{equation}
\lim_{k \to \infty} x'P_k(\gamma^*(k))x = x'P_\infty(\gamma_\infty^*)x \quad \text{for all } x \in \bbR^n
\label{eq:quad-form-conv}
\end{equation}

To establish matrix convergence from quadratic form convergence, define $Z_k \eqbyd P_k(\gamma^*(k)) - P_\infty(\gamma_\infty^*)$, so each $Z_k$ is symmetric and \eqref{eq:quad-form-conv} gives $x'Z_k x \to 0$ for all $x \in \bbR^n$. The polarization identity for symmetric matrices states that for any $x, y \in \bbR^n$ we have
\[
x'Z_k y = \frac{1}{4}\left[(x+y)'Z_k(x+y) - (x-y)'Z_k(x-y)\right]
\]
Since $(x+y)'Z_k(x+y) \to 0$ and $(x-y)'Z_k(x-y) \to 0$ by \eqref{eq:quad-form-conv}, we obtain $x'Z_k y \to 0$ for all $x, y \in \bbR^n$. In particular, for the standard basis vectors $e_i, e_j \in \bbR^n$, we have $(Z_k)_{ij} = e_i'Z_k e_j \to 0$, so every entry of $Z_k$ converges to zero. In finite dimension, entrywise convergence implies convergence in any matrix norm, hence $\norm{Z_k} \to 0$ and we conclude for each fixed $x_0$
\[
\lim_{k\to\infty} P_k(\gamma^*(k)) = P_\infty(\gamma_\infty^*)
\]

\textbf{Steady-state convergence.} \newline
Given that $\lim_{k\to\infty}\gamma^*(k)
      \;=\;
\gamma^*_{\infty}$ and $\lim_{k\to\infty}P_k(\gamma^*(k)) \;=\; P_\infty(\gamma_\infty^*)$, take the limit of recursion \eqref{reverse}
\begin{equation}
\begin{split}
P_\infty(\gamma^*_{\infty}) &= Q+A'P_\infty A- A' P_\infty \begin{bmatrix} B & G \end{bmatrix} \\
&\quad \begin{bmatrix} B'P_\infty B + R & B'P_\infty G \\ (B'P_\infty G)' & G'P_\infty G - \gamma^*_{\infty} I  \end{bmatrix}^{-1}
    \begin{bmatrix} B' \\ G'\end{bmatrix}P_\infty A
\end{split}
\end{equation}
By construction, $\gamma^*(k) \geq \norm{G'P_{k}(\gamma^*(k))G}$ for every $k$, since this is the feasibility constraint. Taking limits and using continuity of $P_k(\cdot)$ in both arguments we obtain
\[
\gamma_\infty^* = \lim_{k\to\infty}\gamma^*(k) \geq \lim_{k\to\infty}\norm{G'P_{k}(\gamma^*(k))G} = \norm{G'P_\infty(\gamma_\infty^*) G}
\]
where the equality holds by continuity of the operator norm. Furthermore, passing to the limit in the Riccati recursion establishes that $(\gamma_\infty^*, P_\infty(\gamma_\infty^*))$ satisfies $g(\gamma_\infty^*, P_\infty(\gamma_\infty^*)) = 0$. 
Therefore, since $\gamma_\infty^* \geq \norm{G'P_\infty(\gamma_\infty^*) G}$ and $g(\gamma_\infty^*, P_\infty(\gamma_\infty^*)) = 0$, $(\gamma_\infty^*, P_\infty(\gamma_\infty^*))$ is a solution to the steady-state problem \eqref{eq:ss-problem}.

We now establish that $(\gamma_\infty^*, P_\infty(\gamma_\infty^*))$ is a minimal solution to the steady-state problem \eqref{eq:ss-problem}. 
Applying \cref{receq} to $g(\overline{\lambda},\overline{\Pi})=0$ we obtain that
\begin{align*}
&Q+A'\overline\Pi A -A' \overline\Pi \begin{bmatrix} B & G \end{bmatrix} \\
&\quad \begin{bmatrix} B'\overline\Pi B + R & B'\overline\Pi G \\ (B'\overline\Pi G)' & G'\overline\Pi G - \overline\lambda I  \end{bmatrix}^{-1}
\begin{bmatrix} B' \\ G'\end{bmatrix}\overline\Pi A - \overline\Pi=0
\end{align*}
which can be rewritten as
\[
\bar{Q} + \bar{A}'\overline\Pi \bar{A} - \bar{A}' \overline\Pi G\,(G'\overline\Pi G-\overline\lambda I)^{\dagger}G'\overline\Pi \bar{A}-\overline\Pi =0,
\]
where $\bar{A} = A+B\overline{K}$, $\bar{Q} = Q + \overline{K}'R\overline{K}$, and $\overline{K}$ satisfies
\[
\begin{bmatrix} B'\overline\Pi B + R & B'\overline\Pi G \\ (B'\overline\Pi G)' & G'\overline\Pi G - \overline\lambda I  \end{bmatrix}\!\begin{bmatrix}
\overline{K} \\ \overline{J}
\end{bmatrix} =\begin{bmatrix} B' \\ G'\end{bmatrix}\overline\Pi A.
\]
Note that since $M_k(\lambda)$ is invertible on $[\lambda_k,\infty)$ under Assumption 2, the Moore–Penrose pseudoinverse in \cref{receq} reduces to the standard inverse.
Since $\overline{\Pi}\succeq0$ and $G'\overline\Pi G-\overline\lambda I\preceq 0$, we have 
$-\bar{A}' \overline\Pi G\,(G'\overline\Pi G-\overline\lambda I)^{\dagger}G'\overline\Pi \bar{A}\succeq0$, hence
\[
\overline{\Pi} \;=\; \bar Q + \bar A'\overline\Pi \bar A - \bar{A}' \overline\Pi G\,(G'\overline\Pi G-\overline\lambda I)^{\dagger}G'\overline\Pi \bar{A} \ \succeq\ \bar Q \ \succeq\ Q \ \succ\ 0
\]
Therefore, because $\overline{\Pi} \succeq P_0$ and the monotone nondecreasing recursion starts at $P_0$, the limit of the recursion gives a minimal positive semidefinite solution to the steady-state problem. Thus, $\overline{\Pi}=P_\infty(\gamma^*_\infty)$, $\overline{\lambda}=\gamma^*_{\infty}$, where $(\overline{\lambda},\overline{\Pi})$ is a minimal solution to the steady-state problem.

    Finally, we establish the correspondence between the forward indexing used in the proof and the backward indexing used elsewhere. For a fixed horizon $N$, identify the scalar multipliers via $\gamma=\lambda$ and map
\[
P_k(\gamma)=\Pi_{N-k}(\lambda)\qquad k=0,1,\dots,N\quad\text{with } P_0=\Pi_N=P_f
\]
In particular, $P_N(\gamma)=\Pi_{0}(\lambda)$ is the stage-0 value matrix. Hence, since $\gamma^*(N)=\lambda^*(N)$ under this identification we have
\[
\lim_{N\to\infty} P_N(\gamma^*(N))=P_\infty(\gamma_\infty^*)\quad\Longleftrightarrow\quad
\lim_{N\to\infty} \Pi_{0}(\lambda^*(N))=\overline{\Pi}
\]
\[
\lim_{N\to\infty}\gamma^*(N)=\gamma_\infty^*\quad\Longleftrightarrow\quad
\lim_{N\to\infty}\lambda^*(N)=\overline{\lambda}
\]
as stated in the proposition.
\end{proof}

\begin{proof}[Proof of \cref{infld}]
To simplify notation, where appropriate, we reverse the time indexing exactly as in \cref{infl}.
\begin{enumerate}
    \item Consider any degenerate system. The result 
\[
  \lim_{N \to \infty} \lambda^*(N) = \lambda^*_\infty \quad \text{and} \quad \lim_{N \to \infty} \Pi_{0}(\lambda^*(N)) = \Pi_\infty(\lambda^*_\infty)
\]
is proved exactly as in \cref{infl}.
\item Consider $x_0\in\mathcal{X}^\infty_{L}(\alpha)$. By definition of $\mathcal{X}^\infty_{L}(\alpha)$, the admissibility constraint is active along the finite horizon recursion. In forward time this means there exists $\hat{k}$ such that $\gamma^*(k)\!=\!\norm{G'P_{k}(\gamma^*(k))G}$ for all $k \geq \hat{k}$.
But for a degenerate steady-state solution $(\overline{\lambda},\overline{\Pi})$, we have
$\overline{\lambda}>\norm{G'\overline{\Pi}G}$ for all $x$, so $(\gamma_\infty^*,P_\infty(\gamma_\infty^*))\neq(\overline{\lambda},\overline{\Pi})$.

Moreover, by monotonicity of the recursion \eqref{reverse} we have
\[
P_{k+1}(\gamma)-P_k(\gamma)\;=\;g\!\big(\gamma,P_k(\gamma)\big)\ \succeq\ 0.
\]
Evaluating at the optimizer $\gamma=\gamma^*(k)$ and letting $k\to\infty$, with
$\gamma^*(k)\!\to\!\lambda^*_\infty$ and $P_k(\gamma^*(k))\!\to\!\Pi_\infty(\lambda^*_\infty)$, continuity yields
$g(\lambda^*_\infty,\Pi_\infty(\lambda^*_\infty))\succeq 0$.
If $g(\lambda^*_\infty,\Pi_\infty(\lambda^*_\infty))=0$ held, then, together with the boundary identity
$\gamma^*(k)\!=\!\norm{G'P_{k}(\gamma^*(k))G}$ for all $k$, this would produce a steady-state fixed point at equality,
contradicting degeneracy, which requires the strict slack
$\overline{\lambda}>\norm{G'\overline{\Pi}G}$ for any steady-state solution.
Therefore $g(\lambda^*_\infty,\Pi_\infty)\succ 0$.

\item Consider $x_0\in\mathcal{X}^\infty_{NL}(\alpha)$. Then there exists $\hat{k}$ such that for all $k\ge \hat{k}$ the admissibility condition is strict at the optimizer: 
$\gamma^*(k)>\norm{G'P_{k-1}(\gamma^*(k))G}$. Passing to the limit, as in \cref{infl}, and using continuity, the Riccati recursion converges and satisfies
\[
g(\gamma_\infty^*,P_\infty(\gamma_\infty^*))=0\qquad 
\gamma_\infty^*\ >\ \norm{G'P_\infty(\gamma_\infty^*)G}
\]
By the same arguments in \cref{infl}, $(\gamma_\infty^*,P_\infty(\gamma_\infty^*))$ is a minimal steady-state solution. Denoting this pair by $(\overline{\lambda},\overline{\Pi})$ gives the claim for $x_0\in\mathcal{X}^\infty_{NL}(\alpha)$.
\end{enumerate}
\end{proof}

\bibliographystyle{abbrvnat}
\bibliography{paper-ib_arxiv}

\end{document}